\documentclass[sn-mathphys,Namedate]{sn-jnl}
\usepackage{graphicx}%
\usepackage{multirow}%
\usepackage{amsmath,amssymb,amsfonts}%
\usepackage{amsthm}%
\usepackage{mathrsfs}%
\usepackage[title]{appendix}%
\usepackage{xcolor}%
\usepackage{textcomp}%
\usepackage{manyfoot}%
\usepackage{booktabs}%
\usepackage{algorithm}%
\usepackage{algorithmicx}%
\usepackage{algpseudocode}%
\usepackage{listings}%
\usepackage{amscd,amsthm,enumerate}
\usepackage{epsfig}
\numberwithin{equation}{section}
\usepackage{caption}
\usepackage{subcaption}
\usepackage{setspace}
\usepackage[final]{changes}
\usepackage{dsfont}
\usepackage{float,amsthm,amssymb,amsfonts}
\usepackage{tikz}
\usepackage{epsf}
\newtheorem{theorem}{Theorem}

\newtheorem{remark}{Remark}[section]
\newtheorem{lemma}{Lemma}

\def\tablename{\bf Table}%
\def\figurename{ \footnotesize Figure}%

\newcommand{\calC}{\mathcal{C}}
\newcommand{\maxcl}{\mathcal{C}_{\infty}}
\newcommand{\calE}{\mathcal{E}}
\newcommand{\calEst}{\mathcal{E}_*}
\newcommand{\calG}{\mathcal{G}}
\newcommand{\calGst}{\mathcal{G}_*}
\newcommand{\calN}{\mathcal{N}}
\newcommand{\calT}{\mathcal{T}}
\newcommand{\calTst}{\mathcal{T}_*}
\newcommand{\calV}{\mathcal{V}}
\newcommand{\calVst}{\mathcal{V}_*}

\newcommand{\bxi}{\boldsymbol{\xi}}

\begin{document}

\title{Modelling the spread of two successive SIR epidemics on a configuration model network}

\author[1]{\fnm{Frank} \sur{Ball}}

\author[2,3]{\fnm{Abid Ali} \sur{Lashari}}

\author[1]{\fnm{David} \sur{Sirl}}

\author*[2,4]{\fnm{Pieter} \sur{Trapman}}\email{j.p.trapman@rug.nl}

\affil[1]{\orgdiv{ School of Mathematical Sciences}, 
\orgname{University of Nottingham, UK}}

\affil[2]{\orgdiv{ Department of Mathematics}, \orgname{Stockholm University, Sweden}} 

\affil[3]{
\orgname{Karolinska Institutet}, \orgaddress{ \city{Stockholm, Sweden}}}

\affil[4]{\orgdiv{Bernoulli Institute}, \orgname{University of Groningen, Netherlands}}


\abstract{We present a stochastic model for two successive SIR (Susceptible $\to$ Infectious $\to$ Recovered) epidemics in the same network structured population. Individuals infected during the first epidemic might have (partial) immunity for the second one.
The first epidemic is analysed through a bond percolation model, while the second epidemic is approximated by a three-type branching process in which the types of individuals depend on their position in the percolation clusters used for the first epidemic. This branching process approximation enables us to calculate, in the large population limit and conditional upon a large outbreak in the first epidemic, a threshold parameter and the probability of a large outbreak for the second epidemic. A second branching process approximation enables us to calculate the fraction of the population that are infected by such a second large outbreak.

We illustrate our results through some specific cases which have appeared previously in the literature
and show that our asymptotic results give good approximations for finite populations.}


\maketitle
\section{Introduction}

In the first half of 2020 many countries and regions experienced an outbreak of the severe acute respiratory syndrome coronavirus 2 (SARS-CoV-2) pandemic, which causes coronavirus disease 2019 (COVID-19).  It took several months of often heavy restrictions for many countries to reduce the prevalence to low levels.  After this, many restrictions were relaxed and a second wave of the pandemic hit those countries and regions, spreading among people who were either still susceptible or had become susceptible again after the first wave.
Some other diseases, such as influenza, follow seasonal patterns, and often leave part of the population which was infected in one outbreak (partially) immune to an epidemic in the next year. For yet other diseases, having been infected in the past or currently being infected with one disease might give increased susceptibility for the second disease (see e.g.\ the introductions of \cite{Moore} and \cite{Bans12} and references therein).

The modelling of this kind of interacting stochastic epidemics on (social) networks has been the subject of some recent epidemiological research. Most models assume that the second epidemic spreads after the first epidemic has finished spreading, so the only relevant information from the first epidemic is which individuals were infected during the epidemic, not when they got infected.  In this paper we study the case in which two epidemics spread successively in a single population modelled by a configuration model network (e.g.\ \cite[Chapter 3]{Durr07} or \cite[Chapter 7]{van2016random}) in which individuals are represented by vertices and relationships which are relevant for the spread of the disease by edges.
We assume that a first disease spreads through the network as an SIR (Susceptible $\to$ Infectious $\to$ Recovered) epidemic, which causes a major outbreak infecting a nonnegligible fraction of the population.
The spread of this epidemic is modelled using bond percolation (see e.g.~\cite{Grimmett1999, Kuul82, newman2002spread} and references therein).
The connected percolation cluster of vertices of the initial infected vertex is distributed as the set of individuals that are no longer susceptible at the end of the epidemic \citep{cox1988limit}. A large outbreak then corresponds to the percolation cluster containing the initial infected vertex, being of the same order as the population size (for a detailed description see \cite{Ball09,Meester}).

The first epidemic continues until there is no infectious individual remaining in the population. After that, the second SIR epidemic spreads through the same network, where the transmission probabilities between given individuals depend on whether or not those individuals were infected in the first epidemic. We assume that the first epidemic results in a large outbreak and obtain a threshold parameter which determines whether a large second epidemic is possible, together with expressions for the probability of a large second epidemic and the fraction of the population that are infected by a large second epidemic. These expressions are asymptotically exact in the limit as the population size $n \to \infty$, unlike much closely related previous work in this area.

SIR epidemics in a closed population are among the most studied class of stochastic epidemic models. In such models, an individual is in one of three possible states, susceptible $(S)$, infected $(I)$ and recovered $(R)$, and the only possible transitions are from $S$ to $I$ and from $I$ to $R$.  A first quantity of interest in the SIR model is the basic reproduction number $R_{0}$ \cite[Chapter 7]{diekmann2012mathematical}, which is a threshold parameter that determines whether or not major outbreaks infecting a strictly positive fraction of the population occur with non-zero probability as $n\to \infty$. It is usually defined as the average number of secondary infections caused by a typical newly-infected individual in the early stages of an epidemic started by one or few initially infectious individuals in an otherwise susceptible population. In line with much of stochastic epidemic theory \citep{Ande00}, one can use branching process approximations to prove a threshold theorem 
which determines whether or not, in the large population limit,
there is a non-zero probability for the epidemic to infect a strictly positive fraction of the population.
The branching process is subcritical, critical or supercritical if $R_{0}<1$, $R_{0}=1$ or $R_{0}>1$, respectively. These methods also yield
approximations for the probability of a large outbreak (i.e.\ the probability that a supercritical epidemic will grow large) and, using
closely related methods, one can also study final outcome (i.e.\, fraction of the population that is infected during an epidemic) properties of such a large outbreak \citep{Ball02,Ball09,britton2010stochastic,Kenah11}.

We use similar methods to analyse the second epidemic given that the first epidemic resulted in a large outbreak.  We approximate the early spread of such a second epidemic using a three-type branching process in which the types of individuals depend on the first epidemic (or, more precisely, on their position in the corresponding percolation component).  This approximation yields a threshold parameter for such a second epidemic and also the approximate probability of a large second epidemic.  We use a closely related three-type branching process to obtain an approximation to the fraction of the population that are infected by a large second epidemic.  A disadvantage of using bond percolation to analyse the first epidemic is that it requires the infectious period to be constant.  Our methodology can be extended to the case when the infectious periods are independent and identically distributed, with an arbitrary but specified distribution. However, the arguments are more complex and four-type approximating branching processes are required.  For ease of presentation we present our results within the framework of a fixed infectious period.  We outline how to extend our analysis to allow for random infectious periods in the discussion in Section~\ref{sec:disc}.

A number of authors have previously considered the spread of more than one epidemics on a network.  \cite{newman2005threshold} studied the existence of a threshold for a second epidemic which spreads on the same configuration model network
as a first epidemic, if the individuals that were infected in the first epidemic are completely immune to the second one.
\cite{Bans12} work in the same framework but relax the assumption of the first epidemic giving complete immunity to the second epidemic, so instead it confers partial immunity.  However, as we show below, the analysis in \cite{Bans12} is based on implicit approximation\added{s}, which are not exact in the large population limit.
\cite{pmid23951134} consider the case in which only individuals infected during the first epidemic are susceptible to the second epidemic.
\cite{Moore} consider a similar model on an Erd\H{o}s-R{\'e}nyi random graph (see e.g.\ \cite[Chapter 2]{Durr07}, \cite[Chapters 4 and 5]{van2016random}), but in their model the second epidemic can only infect individuals which are still infectious in the first epidemic\comment{This is actually correct, it is about superinfection in Moore}. So, in contrast to the models in \cite{newman2005threshold,Bans12,pmid23951134}, Moore et al.\ have to keep track of the spread of the first epidemic as it evolves in time in order to study the spread of the second epidemic.
\cite{Funk10} consider the spread of two epidemics in the same population, where the set of connections between individuals are different for the two epidemics and are generated as two possibly dependent configuration models. In their model the first epidemic provides immunity against the second epidemic and they analyse how the spread of the second epidemic is influenced by the dependence between the two networks.

Our model is closely related to the models presented by \cite{pmid23951134} and \cite{Bans12}.
\cite{Bans12} present two models involving partial immunity. First, they model {\it polarized} partial immunity by assuming that a fraction of the individuals infected in the first epidemic are immune to the second epidemic while all other individuals are fully susceptible for the second epidemic.
They then consider the number of neighbours (individuals connected by an edge in the configuration model) which are susceptible to the second epidemic, of a uniformly chosen individual among the individuals susceptible for the second epidemic and create a new configuration model
of those individuals susceptible to the second epidemic.
Using the standard configuration model and bond percolation on it, they derive epidemiological quantities for a second outbreak in the partly immunised population.
Following arguments in \cite{pmid23951134}, we note that this approach leads to an approximation which is \emph{not} exact in the large population limit.
This can be seen by considering the individuals infected during the first epidemic.
Those vertices necessarily form a connected subgraph. However, in a newly constructed configuration model, they do not need to be all connected.

In the second model of \cite{Bans12}, {\it leaky} partial immunity, which reduces the infectivity and susceptibility of all individuals infected in the first epidemic by the same factor, is considered. In Bansal and Meyers' analysis of this model, a two-type bond percolation process is used.
As for polarized partial immunity, this leads to a tractable analysis which permits exploration of the properties of the approximating model.  However, it too is \emph{not} exact in the large population limit, again because the vertices that are infected in the first epidemic necessarily form a connected subgraph, a property which is lost in the approximation of \cite{Bans12}.


This paper is organised as follows. In Section \ref{sec:epidemic}, we describe our epidemic model. We introduce the configuration model graph/network and the two consecutive SIR epidemics defined on this graph. We describe the so-called polarized and leaky partial immunity models.  In Section \ref{sec:Results}, we state the main results of our paper in four theorems.   Theorems~\ref{thmR0} and~\ref{thmsurvival} give a threshold parameter for a second large outbreak to have strictly positive probability and the computation of the probability of such a large outbreak, respectively.   Theorem~\ref{thmfinalsize} gives the computation of the fraction of the population infected by a large second epidemic. Leaky partial immunity is a special case of our more general model for the second epidemic, so these three theorems apply to that model. Theorem~\ref{thmpol} gives equivalent results to the first three theorems but for polarized partial immunity, which does not fit exactly our framework.  Computation of the threshold parameter and large outbreak probability for the second epidemic is the same as before but computation of the fraction infected by a large epidemic needs modification.
In Section \ref{examples}, we analyse three examples.  The first concerns the herd immunity threshold after the first epidemic when individuals who are infected by it are immune to the second epidemic.  The second and third correspond to the modelling of polarized and leaky partial immunity as introduced by \cite{Bans12}.  In Section \ref{sec:illus} we show through numerical simulations that our asymptotic results perform well in finite networks.
In Section \ref{sec:proofs} we provide the proofs of the main theorems. To do this we first present an approximating three-type branching process for the second epidemic and compute its mean offspring matrix and its  probability of survival, which corresponds to the probability of a large outbreak in the second epidemic. Using another approximating three-type branching process, we also give a recipe to compute the fraction of the population infected in such a large outbreak.     Finally, Section \ref{sec:disc}, we discuss some possible extensions, including non-constant infectious periods, and other future work.

\section{The epidemic model}
\label{sec:epidemic}
\subsection{The configuration model}
\label{sec:conmod}
Formally, we consider $\{E^{(n)};{n \in \mathbb{N}}\}$, a sequence of models for the spread of SIR epidemics on random graphs/networks. To be specific, we assume that the models are independent for different $n$. The epidemic $E^{(n)}$ spreads on the random graph
$\calG= \calG^{(n)} = (\calV^{(n)}, \calE^{(n)})$, where the vertex set $\calV^{(n)}$ consists of $n$ vertices that represent the individuals.  The edges in the edge set $\calE^{(n)}$ represent connections/relationships between individuals through which infection might transmit.
We obtain results for the asymptotic case $n \to \infty$. In the remainder of the paper we often suppress the superscript $(n)$.

The graph $\calG$ is generated by (a version of) the configuration model (for a detailed description see \cite[Chapter 3]{Durr07} or \cite[Chapter 7]{van2016random}).
We construct $\calG$ by assigning an i.i.d. (independent and identically distributed) number of half-edges to each vertex.
The number of half-edges assigned to a vertex is called its degree. The degrees are distributed according to an arbitrary but specified discrete random variable $D$ having probability mass function
$$p_k= \mathbb{P}(D = k), \qquad (k = 0, 1,...).$$
We assume that $\mathbb{E}[D^{2}]<\infty$. The half-edges are then paired uniformly at random to create
the edges in the graph.
In this construction some imperfections might occur.
If the total number of half edges is odd, we ignore the single leftover half-edge which has no effect on the degree distribution as $n\to\infty$.
Furthermore,  it is well-known that the numbers of self-loops (edges from a vertex to itself) and multiple edges between the same pair of vertices are small (the numbers are bounded in expectation as $n \to \infty$) under our assumption that  $\mathbb{E}[D^{2}]<\infty$ \cite[p.\ 230]{van2016random}. Moreover, the probability that the graph has no such imperfections is bounded below by a strictly positive constant as $n \to \infty$, so convergence in probability results may be transferred from epidemics on the constructed graph $\calG$ to those on $\calG$ conditioned on having no such imperfections \citep{Jans14}. Therefore we can safely ignore self-loops and merge parallel edges.

In the construction of the configuration model, a half-edge is $k$ times as likely to be paired with a half-edge emanating from
an individual with degree $k$ than with one emanating from an individual with degree 1. Therefore, a typical \added{(i.e.\ uniformly chosen)} neighbour of a typical individual in the network has a size-biased degree distribution \citep{newman2002spread} defined through
\begin{equation}
\label{sizebias}
\tilde{p}_{k}=\mathbb{P}(\tilde{D}=k)=\frac{1}{\mu_{D}}kp_k,\qquad(k=1,2,\cdots),
\end{equation}
where
\begin{equation}
\label{Dmean}
\mu_{D}=\sum_{\ell=0}^{\infty}\ell p_{\ell}
\end{equation}
is the expected degree of a vertex.

\subsection{The first epidemic}
\label{sec:epigra}

\added{In this subsection we discuss how the first epidemic spreads on a network. For a more extensive discussion and some heuristic and rigorous analysises of those epidemic see e.g.\ \cite[Section 3.5]{Durr07}, \cite{Ande99,Ball09,Kenah07,Mill07,newman2002spread}.}

The first epidemic is defined as follows. We consider an SIR epidemic on $\calG$. Each individual is assumed to be susceptible, infectious or recovered. We say that a vertex is susceptible, infectious or recovered if the individual it represents is in that ``infection state''. Throughout we assume that at time $0$, one individual chosen uniformly at random from the population is infectious, while all other individuals are susceptible. Our model and analysis can easily be generalised to other initial conditions, such as having a bounded (as $n \to \infty$) number of initially infected individuals chosen uniformly at random, or assuming that the initial infectious individual has a specified degree.

Neighbours in the population, i.e.\ individuals/vertices connected by an edge in $\mathcal{G}$, contact each other according to independent homogeneous Poisson processes with rate $\beta$. If the contact is between a susceptible and an infectious vertex, then the susceptible one immediately becomes infectious; otherwise the contact has no effect.
In many papers on stochastic epidemics, infectious individuals stay infectious for a possibly random period distributed as the random variable $L$.
However, in order to improve the readability of the analysis in this paper \added{and avoid having to deal with some depencies (see e.g., \cite{Kuul82,Meester,Kenah11}),}   we assume that $L$ is constant, i.e.\ $L$ is not random. Without loss of generality we assume that $\mathbb{P}(L=1)=1$ (i.e.\ we scale  time such that the infectious period is one time unit). In Section \ref{sec:disc} we remark on the model with more general distributions of $L$.
\added{In particular we discuss why extending the analysis is not \deleted{entirely }trivial, but still explain how this extension can be made.}
Because $L$ is not random, the first epidemic on the network has the same final size distribution as a corresponding Reed-Frost model (see e.g.\ \cite[p.\ 17]{Ande00}).
Because $\mathbb{P}(L=1)=1$, any given infected individual makes at least one contact during its infectious period with each of its neighbouring individuals independently and with the same probability, $p=1-\rm{e}^{-\beta}$.
The first epidemic ends when there is no infectious individual remaining in the population.

We define the basic reproduction number of the first epidemic as
\begin{equation}
\label{firstR0}
R_0^{(1)}= p \mathbb{E}[\tilde{D}-1]=p\left(\mu_{D}+\frac{\sigma^{2}_{D}}{\mu_{D}}-1\right),
\end{equation}
where $\tilde{D}$ is defined in \eqref{sizebias}, $\mu_D$ in \eqref{Dmean} and $\sigma^{2}_{D}=\sum_{\ell=0}^{\infty}\ell^{2} p_{\ell}-\mu_{D}^{2}$ is the variance of $D$.
Note that $\mathbb{E}[\tilde{D}-1]$ is the expected number of susceptible neighbours a typical vertex has just after being infected in the early stages of the epidemic, while each of those susceptible neighbours is infected by this vertex with probability $p$. So, $R_0^{(1)}$ is the expected number of neighbours infected by a typical infected vertex during the early stages of an epidemic.
It is well-known (e.g.\ \cite{newman2002spread,britton2010stochastic}) that this quantity is a threshold parameter for the first epidemic: in a large population, a large outbreak occurs with strictly  positive probability if and only if $R_0^{(1)}>1$.

It follows from branching process approximations, see e.g.~\cite[Theorem 3.5.1]{Durr07}), \cite{Kenah07}, or \cite{Ball09}, that both  the probability of a large outbreak and, conditioned on a large outbreak, the fraction of the population infected  converge (as $n \to \infty$) to the probability that a well-chosen 2-stage Galton-Watson process survives.  This probability can be shown to  be $1-q$, where  the constant $q$ is given by
\begin{equation}
\label{qdef0}
q=\mathbb{E}[(1-p+p \tilde{q})^{D}],
\end{equation}
where $\tilde{q}$ is the smallest positive solution of
\begin{equation}
\label{qdef}
s=\mathbb{E}[(1-p+ps)^{\tilde{D}-1}].
\end{equation}
Indeed, it is easily checked that $\tilde{q} \in [0,1)$ (and hence $q \in [0,1)$) if and only if $R_0^{(1)}>1$, excluding the pathological case when
all vertices have degree $2$ and $p=1$. The equality of the probability of a large outbreak and the fraction of the population infected in a large outbreak follows  immediately if one uses bond percolation arguments (see e.g.\ \cite{cox1988limit,Kenah11,Meester} and Section \ref{percapproach} below) to study the epidemic, and critically depends on the assumption that the infectious period is not random.

\subsection{The second epidemic}
\label{subsec:secondepi}
To analyse the second epidemic we keep track of whether or not vertices were infected in the first epidemic. So, we need four rates $\beta_{00}$,  $\beta_{01}$,  $\beta_{10}$ and  $\beta_{11}$ to denote the rates of infectious contacts between different types of vertices. The parameter $\beta_{00}$ is the rate at which an individual who was not infected during the first epidemic makes infectious contacts with a given neighbour who was also not infected during the first epidemic. Here infectious contacts need not be symmetric and are defined as  contacts which would lead to infection if the ``contacter'' is infectious and the ``contactee'' is susceptible.
Likewise, $\beta_{01}$ is the rate at which an individual who was not infected during the first epidemic makes infectious contacts  with a given neighbour who was infected during the first epidemic. The rates
$\beta_{10}$ and  $\beta_{11}$ are defined similarly.
Furthermore, we assume that the infectious period of the disease spreading in the second epidemic for an $i$-individual ($i \in \{0,1\}$) is fixed at length $\ell_i^{(2)}$ and define
\begin{equation}
\label{pidef}
\pi_{ij} = 1-{\rm e}^{-\beta_{ij} \ell_i^{(2)}}, \qquad \mbox{for $i,j \in \{0,1\}$},
\end{equation}
that is, $\pi_{ij}$ is the probability that, if it gets infected, an $i$-individual makes infectious contact with a given $j$-neighbour.

\begin{remark}\label{Rem1}
It is easy to check that if $\pi_{11}=p'$ for some $p'\in [0,1]$ and $\pi_{10}=\pi_{01}=\pi_{00}= 0$ then our model corresponds to the model in \cite{pmid23951134}. Moreover, if we take $\pi_{00}=p'$ for some $p'\in [0,1]$ and $\pi_{11}=\pi_{01}=\pi_{10}= 0$ then our model corresponds to the model in \cite{newman2005threshold}. Thus, the models in \cite{pmid23951134} and \cite{newman2005threshold} are special cases of our model.
\end{remark}

\subsection{Models for partial immunity}
\label{sec:partial}

In this subsection we describe the models for polarized and leaky partial immunity as introduced by \cite{Bans12} and outline their relationship to the model for the second epidemic in Section~\ref{subsec:secondepi}.

\subsubsection{Polarized partial immunity}
\label{sec:poldef}

In the polarized partial immunity model, every recovered individual after the first epidemic becomes susceptible to the second epidemic independently with probability $\alpha$, while the other recovered individuals stay immune.   Note that this model does not fit completely the framework in Section~\ref{subsec:secondepi}\replaced{. Indeed,}{ because} the susceptibility of an individual to the second epidemic depends not only on whether or not it was infected by the first epidemic\added{,} but also on whether it is rendered immune to the second epidemic if it was infected by the first epidemic\replaced{. The latter}{, which} is random unless $\alpha=0$ or $\alpha=1$.\comment{PT Oct11:rewrote because I felt sentence was too long}  Our results concerning the probability of a large second epidemic are based upon a branching process approximation to the early stages of the second epidemic, a key assumption of which is that, apart from its infector, every individual contacted by an infective is someone who has not been contacted previously.  Thus, in that branching process approximation, we may reduce the probability of transmission to any individual infected in the first epidemic by the factor $\alpha$ and the asymptotic probability of a large second epidemic can be obtained by setting
\begin{equation}
\label{pi-pol}
\pi_{00}= p', \pi_{01}=\alpha p', \pi_{10}=p' \text{ and } \pi_{11}=\alpha p'
\end{equation}
in our general theory.  (Here $p'$ is the probability that in the second epidemic an infective infects a given susceptible neighbour who was not infected by the first epidemic, which can be different from $p$, the infection probability in the first epidemic.)  However, it is \emph{incorrect} to do this when considering the size of a large second epidemic because once the epidemic has become established, a given individual who was infected in the first epidemic may be contacted by more than one distinct infective in the second epidemic and the outcomes of those contacts are \emph{not} independent (unless $\alpha=0$ or $\alpha=1$), since they depend on whether or not the individual is immune to the second epidemic.  Nevertheless, our proof can be adapted
to determine the asymptotic final size of a large second epidemic under this model for polarized partial immunity.

\subsubsection{Leaky partial immunity}
\label{sec:leakydef}

Recall that under leaky partial immunity all individuals infected by the first epidemic are affected identically in terms of the second epidemic; they have reduced susceptibility and also reduced infectivity if they are infected.  \cite{Bans12} model this by assuming that the probability of transmission per susceptible-infectious neighbour pair is reduced by a factor $\alpha$ if one of the pair was infected by the first epidemic and by a factor $\alpha^2$ if both were infected by the first epidemic.  This can be generalised by having factors $\alpha_S$ and $\alpha_I$ for susceptibility and infectivity,  respectively, and fits our general theory with
\begin{equation}
\label{pi-leaky1}
\pi_{00}=p', \pi_{01}=\alpha_S p', \pi_{10}=\alpha_I p' \text{ and } \pi_{11}=\alpha_S \alpha_I p'.
\end{equation}

An alternative, which aligns more closely with models for vaccine action is to assume that these reductions are in terms of changed  probability of transmission per contact between a susceptible and infectious individual, depending on whether the individuals involved have been infected before.  To model this, give individuals in the second epidemic a further typing; specifically call individuals that were infected by the first epidemic type $V$ and those that were not infected type $U$.  (This labelling reflects the similarity between immunity conferred by prior infection to that conferred by vaccination.)  Then, in the second epidemic, each contact between an infectious and a type-$V$ susceptible leads to the infection of the latter independently with probability $\alpha_S$, and the rate at which a type-$V$ infectious contacts
its neighbours is a factor $\alpha_I$ lower than the corresponding rate for a type-$U$ infectious.  This implies that
the $\beta_{ij}$s defined just before~\eqref{pidef} take the form $\beta_{00}=\beta', \beta_{01}=\alpha_S\beta',
\beta_{10}=\alpha_I\beta'$ and $\beta_{11}=\alpha_S \alpha_I \beta'$, whence \eqref{pidef} now yields
\begin{equation}
\label{pi-leaky2}
\pi_{00}= 1-\rm{e}^{-\beta'\ell_0^{(2)}},
\pi_{01}=1-\rm{e}^{-\beta'\alpha_{S}\ell_0^{(2)}}, \pi_{10}=1-\rm{e}^{-\beta'\alpha_{I}\ell_1^{(2)}} \text{and}
\pi_{11}=1-\rm{e}^{-\beta'\alpha_{S}\alpha_{I}\ell_1^{(2)}}.
\end{equation}

Note that the contact rate $\beta'$ between individuals not infected by the first epidemic need not be the same as the contact rate $\beta$ in the first epidemic.

\begin{remark}\label{remleak1}
 We note that here (and in \cite{Bans12}) the term leaky has a slightly different meaning than it usually has in vaccine response models, where it is usually assumed that infectivity is not affected by the vaccine (i.e.\ $\alpha_I=1$).  Our model for leaky partial immunity is analogous to the non-random vaccine response of \cite{BallLyne06},
which is a special case of the vaccine action model introduced by \cite{BeckerStarczak98}.
\end{remark}

\section{Results}\label{sec:Results}

In this section we present the main results of the paper, which provide for the second epidemic (i) a threshold parameter which determines whether a large outbreak is possible, (ii) the probability that a large outbreak occurs and (iii) the fraction of the population that are infected by a large second outbreak.  These results are all conditional upon a large first epidemic, which is assumed implicitly throughout this section.
Strictly speaking, we prove the theorems for the second epidemic on the approximating coloured tree as constructed in Section \ref{sec:coltree}, while we provide a heuristic justification for why using the coloured tree will lead to the corresponding results on a configuration model. Making this justification fully rigorous requires a technical proof similar to proofs in e.g.\ \cite{Ball09,ball_sirl_trapman_2014}, which is beyond the scope of this paper.

To state our first main result we define the following four parameters.
\begin{align}
\mu_{AA}&=\sum_{k=2}^{\infty}(k-1)\tilde{p}_{k}\left(1-(1-p)(1-p+p\tilde{q})^{k-2}\right),\label{eq:4}\\
\mu_{AB}&=\frac{\tilde{q}}{1-\tilde{q}} \sum_{k=2}^{\infty}(k-1)\tilde{p}_{k}\left(1-(1-p+p\tilde{q})^{k-2}\right),\label{eq:5}\\
\mu_{BA}& = \frac{(1-p)(1-\tilde{q})}{\tilde{q}}\displaystyle\sum_{k=2}^{\infty}(k-1)\tilde{p}_{k}(1-p+p\tilde{q})^{k-2},\label{eq:6}\\
\mu_{BB}&= \displaystyle\sum_{k=2}^{\infty}(k-1)\tilde{p}_{k}(1-p+p\tilde{q})^{k-2}\label{eq:7},
\end{align}
where  $\tilde{q}$ is defined as in \eqref{qdef}.
The motivation for the notation is made clear in Section \ref{sec:R0}.

Define the matrix $M$ as follows
\begin{equation}\label{eq:8}
M=\begin{pmatrix}
\pi_{11}\mu_{AA} & p\,\pi_{11}\mu_{AB} &(1-p)\pi_{10}\mu_{AB} \\
\pi_{11}\mu_{BA} & p\,\pi_{11}\mu_{BB} & (1-p)\pi_{10}\mu_{BB} \\
\pi_{01}\mu_{BA} & 0 & \pi_{00}\mu_{BB}
\end{pmatrix} ,
\end{equation}
where for $i,j \in \{0,1\}$, $\pi_{ij}$ is defined in \eqref{pidef} and $p$ is the per pair of neighbours infection probability of the first epidemic. \added{As mentioned in the introduction, we use a three-type branching process to approximate the early spread of the second epidemic.  The matrix $M$ is the mean offspring matrix for that branching process.  The first two types concern individuals that are infected in the first epidemic and the third type individuals that avoided infection in the first epidemic.}  We are now ready to formulate our first main result.

\begin{theorem}\label{thmR0}
Consider the model defined in Section \ref{sec:epidemic}.
Let $R_0^{(2)}$ be the dominant eigenvalue of the matrix $M$. The second epidemic has a strictly positive probability to be major if and only if $R_0^{(2)} >1$.
\end{theorem}

The next result concerns the probability of a major outbreak of the second epidemic. Our results are, as the notation suggests, strongly related to results from the theory of  branching processes.
We use, for  $i,j \in \{1,2\}$, the notation
\begin{align}
t_{ij}(s_j)& =  1-\pi_{11}+\pi_{11}s_j,\nonumber \\ 
t_{i3}(s_3) &=  1-\pi_{10}+\pi_{10}s_3,\nonumber\\
t_{31}(s_1) &=  1-\pi_{01}+\pi_{01}s_1,\nonumber\\
t_{32}(s_2)& = 0,\nonumber\\
t_{33}(s_{3}) & =  1-\pi_{00}+\pi_{00}s_3,\label{tdefs}
\end{align}
to define the following probability generating functions.
For $i\in\{1,3\}$ and $\mathbf{s}=(s_{1},s_{2},s_{3})\in [0,1]^3$ let
\begin{equation}\label{fifirst}
f_{i}(\mathbf{s}) =  f_{\mathbf{Y}_i}(t_{i1}(s_1),t_{i2}(s_2),t_{i3}(s_3)),
\end{equation}
where
$f_{\mathbf{Y}_i}(\mathbf{s})$ $(i \in \{1,3\})$
is given by
\begin{align}
f_{\mathbf{Y}_1} &=\frac{1}{1-q}f_{D}(\left(1-\tilde{q})s_1+\tilde{q}(p s_2+(1-p)s_3\right))\nonumber\\
&\qquad -\frac{1}{1-q}f_{D}(\left(1-p)(1-\tilde{q})s_1+\tilde{q}(p s_2+(1-p)s_3\right)),\label{eq:1}\\
f_{\mathbf{Y}_3}(\mathbf{s})
&= \frac{1}{q}\displaystyle f_{D}\left((1-p)(1-\tilde{q})s_1+\tilde{q}s_3)\right).\label{eq:1.4}
\end{align}
Furthermore,
let
\begin{equation}
\label{equ:PGFXtilde1}
\tilde{f}_{i}(\mathbf{s})=f_{\mathbf{\tilde{Y}}_i}(t_{i1}(s_1),t_{i2}(s_2),t_{i3}(s_3))\quad(i \in \{1,2,3\}, \mathbf{s} \in [0,1]^3),
\end{equation}
where
\begin{align}
f_{\mathbf{\tilde{Y}}_1}(\mathbf{s})
= & \frac{1}{1-\tilde{q}}\displaystyle f_{\tilde{D}-1}((1-\tilde{q})s_1+\tilde{q}(p s_2+(1-p)s_3)) \nonumber\\
 & - \frac{1}{1-\tilde{q}} \displaystyle f_{\tilde{D}-1}((1-p)(1-\tilde{q})s_1+\tilde{q}(p s_2+(1-p)s_3)),\label{eq:2a}\\
f_{\mathbf{\tilde{Y}}_2}(\mathbf{s})
= &\frac{1}{\tilde{q}} \displaystyle f_{\tilde{D}-1}\left((1-p)(1-\tilde{q})s_1+\tilde{q}(p s_2+(1-p) s_3)\right),\label{eq:2b}\\
f_{\mathbf{\tilde{Y}}_3}(\mathbf{s})
= &\frac{1}{\tilde{q}}\displaystyle f_{\tilde{D}-1}\left((1-p)(1-\tilde{q})s_{1}+\tilde{q}s_{3})\right).\label{eq:2c}
\end{align}

\begin{theorem}\label{thmsurvival}
Consider the model defined in Section \ref{sec:epidemic}.
Assume that $R_0^{(2)}>1$ and that the initial infected individual for the second epidemic is chosen uniformly at random.
Let $\bxi=(\xi_{1},\xi_{2},\xi_{3})$ be the unique solution in $[0,1)^{3}$ of the system of equations
\begin{align*}
\xi_1 &= \tilde{f}_1(\bxi),\\
\xi_2 &= \tilde{f}_2(\bxi),\\
\xi_3 &= \tilde{f}_3(\bxi).
\end{align*}
Then the probability of a large outbreak is, asymptotically as $n \to \infty$, given by $$1-q f_{1}(\bxi)-(1-q)  f_{3}(\bxi).$$
\added{Further, the asymptotic probability of a large outbreak is $1-f_{1}(\bxi)$ if the initial infective was infected in the first epidemic and $1-f_{3}(\bxi)$ if it was not.}
\end{theorem}

For the model defined in Section \ref{sec:epidemic} we also obtain results for the fraction of vertices that are ultimately infected by the second epidemic. The calculations for this final outcome are very similar to the calculations of the asymptotic (as $n \to \infty$) probability of a large outbreak if the initial infected individual is chosen uniformly at random.

We use (cf. \eqref{tdefs}, for  $i,j \in \{1,2\}$, the notation
\begin{align}
\hat{t}_{ij}(s_j)& =  1-\pi_{11}+\pi_{11}s_j,\nonumber \\ 
\hat{t}_{i3}(s_3) &=  1-\pi_{01}+\pi_{01}s_3,\nonumber\\
\hat{t}_{31}(s_1) &=  1-\pi_{10}+\pi_{10}s_1,\nonumber\\
\hat{t}_{32}(s_2)& = 0,\nonumber\\
\hat{t}_{33}(s_{3}) & =  1-\pi_{00}+\pi_{00}s_3.\label{tdefs2}
\end{align}

Furthermore define
\begin{equation}
\label{equa:PGFXtilde}
\hat{f}_{i}(\mathbf{s})=f_{\mathbf{\tilde{Y}}_i}(\hat{t}_{i1}(s_1),\hat{t}_{i2}(s_2),\hat{t}_{i3}(s_3))\quad(i \in \{1,2,3\}, \mathbf{s} \in [0,1]^3),
\end{equation}
and
\begin{equation}
\label{equa:PGFXcheck}
\check{f}_{i}(\mathbf{s})=f_{\mathbf{Y}_i}(\hat{t}_{i1}(s_1),\hat{t}_{i2}(s_2),\hat{t}_{i3}(s_3))\quad(i \in \{1,3\}, \mathbf{s} \in [0,1]^3).
\end{equation}

\begin{theorem}\label{thmfinalsize}
Consider the model defined in Section \ref{sec:epidemic}.
Assume that $R_0^{(2)}>1$ and that there is a single initial infected individual for the second epidemic.
Let $\hat{\bxi}=(\hat{\xi}_{1},\hat{\xi}_{2},\hat{\xi}_{3})$ be the unique solution in $[0,1)^{3}$ of the system of equations
\begin{align*}
\hat{\xi}_1 &= \hat{f}_1(\hat{\bxi}),\\
\hat{\xi}_2 &= \hat{f}_2(\hat{\bxi}),\\
\hat{\xi}_3 &= \hat{f}_3(\hat{\bxi}).
\end{align*}
Then the fraction of vertices infected in a large outbreak is, asymptotically as $n \to \infty$, given by $$1-q \check{f}_{1}(\hat{\bxi})-(1-q)  \check{f}_{3}(\hat{\bxi}).$$
\added{Further, the fraction of vertices infected in the first epidemic that are also infected in the second epidemic is $1-\hat{\xi}_1$, and the fraction of vertices not infected in the first epidemic that are infected in the second epidemic is $1-\hat{\xi}_3.$}
\end{theorem}

\added{Recall from Section~\ref{sec:epigra} that the fraction of the population infected by a large first epidemic is $1-q$.  Thus, Theorem~\ref{thmfinalsize} implies that conditional upon both epidemics being large, as $n \to \infty$, the probability an individual chosen uniformly at random from the population avoids both epidemics converges to $q\check{\xi}_3$, the probability it is infected by the first epidemic but not the second converges to $(1-q)\check{\xi}_1$, the probability it is infected by the second epidemic but not the first converges to $q(1-\check{\xi}_3)$ and the probability it is infected by both epidemics converges to $(1-q)(1-\check{\xi}_1)$.}

The two models for leaky partial immunity described in Section~\ref{sec:leakydef} are special cases of the model for the second epidemic in Section~\ref{subsec:secondepi}, so Theorems~\ref{thmR0}-\ref{thmfinalsize} hold for those models, with $\pi_{ij}$ $(i,j \in \{0,1\})$ being given by~\eqref{pi-leaky1}
or~\eqref{pi-leaky2} depending on the model.  As indicated in Section~\ref{sec:poldef} above, the model for polarized partial immunity described there does not fit the framework of Section~\ref{subsec:secondepi}.  However, the proofs of Theorems~\ref{thmR0}-\ref{thmfinalsize} can be adapted to obtain the following result.

\begin{theorem}\label{thmpol}
Consider the model for polarized partial immunity defined in Section~\ref{sec:poldef} and let $\pi_{ij}$ $(i,j \in \{0,1\})$ be given by~\eqref{pi-pol}. Then Theorems~\ref{thmR0} and~\ref{thmsurvival} hold for this model.  Theorem~\ref{thmfinalsize} holds also, provided $\hat{f}_{i}(\mathbf{s})$ $(i \in\{1,2,3\})$ at~\eqref{equa:PGFXtilde} are replaced by
\[
\hat{f}_{i}(\mathbf{s})=1-\alpha+\alpha f_{\tilde{\mathbf{Y}}_i}(1-p'+p's_1,1-p'+p's_2,1-p'+p's_3)\quad(i \in \{1,2\})
\]
and
\[
\hat{f}_{3}(\mathbf{s})=f_{\tilde{\mathbf{Y}}_3}(1-p'+p's_1,1-p'+p's_2,1-p'+p's_3);
\]
and $\check{f}_{i}(\mathbf{s})$ $(i \in \{1,3\})$ at~\eqref{equa:PGFXcheck} are replaced by
\[
\check{f}_1(\mathbf{s})=1-\alpha+\alpha f_{\mathbf{Y}_1}(1-p'+p's_1,1-p'+p's_2,1-p'+p's_3)
\]
and
\[
\check{f}_3(\mathbf{s})=f_{\mathbf{Y}_3}(1-p'+p's_1,1-p'+'ps_2,1-p'+p's_3).
\]
\end{theorem}

\section{Examples: Modelling infection-acquired immunity}
\label{examples}
In this section we use our theory to investigate herd immunity and the potential of a second wave of an epidemic (cf.\ \cite{Brit20}) in a population structured by a configuration model network. We further analyse epidemics with polarized and leaky partial immunity  caused by the first epidemic that were previously studied through an approximation by \cite{Bans12}.

\subsection{The herd immunity threshold after a first outbreak}

When an epidemic spreads through a population, often measures are taken to limit the impact of the disease. If a vaccine is available this can be done through vaccinating (a part of) the population. If a vaccine is not available, non-pharmaceutical interventions have to be taken in order to limit the number of contacts of individuals and/or the probability that a contact leads to transmission of the epidemic.

To model this, we first assume that infection confers complete immunity, so $\pi_{11}= \pi_{01}=0$ and our model reduces to that of \cite{newman2005threshold}.
The second epidemic involves only \replaced{individuals who were  not infected by the first epidemic,}{type-$3$ individuals} and its early stages are approximated by a single-type branching process, whose offspring mean
\deleted{(except for the root)} is $R_0^{(2)} =\pi_{00} \mu_{BB}$ \replaced{(cf.~the bottom-right element of the matrix $M$ defined in~\eqref{eq:8})}
{(cf.~\replaced{Theorem \ref{thmR0}}{\eqref{eq:8})}}.  Thus, using~\eqref{eq:7},
\begin{equation}
\label{eq:R02HI}
R_0^{(2)}=\pi_{00} \sum_{k=2}^{\infty}(k-1) \tilde{p}_k (1-p+p\tilde{q})^{k-2} =  \pi_{00} \mathbb{E}[(\tilde{D}-1)  (1-p+p\tilde{q})^{\tilde{D}-2}],
\end{equation}
where $\tilde{q}$ is the smallest positive solution of \eqref{qdef}.

Now assume that $\pi_{00}$ is the probability that an infectious individual makes an infectious contact with a given neighbour in an unrestricted epidemic and
$$R_0'=\pi_{00}  \mathbb{E}[\tilde{D}-1] >1,$$
so if there was no\added{ large} first epidemic, an unrestricted epidemic is supercritical. We investigate the threshold value of $p$, the infection probability for the first epidemic, for which herd immunity is obtained, that is, the value of $p$ for which $R_0^{(2)}=1$ holds. We then compute the corresponding value of $1-q$, the fraction of the population that is no longer susceptible.
In a homogeneously mixing population \replaced{the latter value is}{both of these values are} $1-1/R^{\text{hom}}_0$, where $R^{\text{hom}}_0$ is the basic reproduction number for the unrestricted epidemic in a homogeneously mixing population \cite[p.~69]{diekmann2012mathematical}.

Note that $R_0^{(2)}$ is non-increasing in $p$ (because increasing $p$ leaves fewer individuals susceptible after the first outbreak). Further, for $p \mathbb{E}[\tilde{D}-1]>1$, it follows from~\eqref{qdef} that $\tilde{q}$ is strictly decreasing and continuous in $p$ and hence, using~\eqref{eq:R02HI}, so is $R_0^{(2)}$.
Let $p'$ be chosen such that $R_0^{(2)} =1$, and let $\tilde{q}'$ and $q'$ be the corresponding $\tilde{q}$ and $q$, respectively. Then we obtain
\begin{equation}\label{HIR02}
1 =\pi_{00} \mathbb{E}[(\tilde{D}-1)  (1-p'+p'\tilde{q}')^{\tilde{D}-2}]
=  \pi_{00} \frac{ \mathbb{E}[D(D-1)  (1-p'+p'\tilde{q}')^{D-2}]}{\mathbb{E}[D]}
\end{equation}
and, using \eqref{qdef0},
\begin{equation*}\label{HIq}
q' =\mathbb{E}[(1-p'+p'\tilde{q}')^{D}].
\end{equation*}
Further, note that $$R'_0= \pi_{00}  \mathbb{E}[\tilde{D}-1] =  \pi_{00}  \frac{\mathbb{E}[D(D-1)]}{\mathbb{E}[D]}.$$

In general it is not possible to give a closed-form expression for $q'$ in terms of ``standard functions''.
However, suppose that $D$ has a Poisson-type distribution, i.e.\ its PGF satisfies $f'_D(s) = f_D'(1) \left(f_D(s)\right)^{\kappa}$ for some $\kappa \in (0,\infty)$ \citep{Jaco18}.
Then, $f_{\tilde{D}-1}(s)=\left(f_D(s)\right)^{\kappa}$ and $f_{\tilde{D}-1}'(s)=\kappa\left(f_D(s)\right)^{2\kappa-1}\mu_D$. Further, it follows from~\eqref{qdef0} and~\eqref{qdef} that $\tilde{q}=q^{\kappa}$, so $f_D(1-p+p\tilde{q})=q^{1/\kappa}$. Hence, $R'_0=\pi_{00}\kappa \mu_D$ and the first equality in~\eqref{HIR02} yields after a little algebra that
\begin{equation}
\label{qPoissontype}
q'=(1/R_0')^{1/(2\kappa-1)},
\end{equation}
so we get an explicit expression for the disease-induced herd immunity level.  Moreover, note that $1-q' < 1-1/R_0'$ if and only if $\kappa>1$.  Thus, the herd immunity threshold obtained through letting a less infectious disease spread through the population is strictly less than the herd immunity level obtained through immunisation of a uniformly chosen subset  of the vertices of an appropriate size if and only if $\kappa >1$ (cf.~\cite{Brit20}).

Some examples of Poisson-type distributions are
(i) $D \sim {\rm Bin}(n, \theta)$, in which case $\kappa= (n-1)/n$, so $q'=(1/R_0')^{n/(n-2)}$,
(ii) $D \sim {\rm NegBin}( r, \theta)$ (with support $\mathbb{Z}_{\geq 0}$), in which case $\kappa= (r+1)/r$, so $q'=(1/R_0')^{r/(r+2)}$,
(iii) $D \sim {\rm Poisson}(\lambda)$, in which case $\kappa=1$, so $q'=1/R_0'$. Note that $\mathbb{P}(D=n)=1$ for some integer $n$, is a special case of $D \sim {\rm Bin}(n, \theta)$, with $\theta=1$.

Returning to the general case, observe that, $1-q' \leq 1-1/R_0'$ if and only if
\begin{equation}\label{Herdmaster}
\mathbb{E}[D(D-1)  (1-p'+p'\tilde{q}')^{D-2}]\leq \mathbb{E}[D(D-1)]\mathbb{E}[(1-p'+p'\tilde{q}')^{D}].\end{equation}
Further, if $D$ has a Mixed-Poisson distribution with mean distribution $X$, i.e.
\[
\mathbb{P}(D=k)=\mathbb{E}[{\rm e}^{-X}X^k]/k! \quad (k=0,1,\dots)
\]
for some non-negative random variable $X$,
then \eqref{Herdmaster} reads
$$\mathbb{E}[X^2 {\rm e}^{-p'(1-\tilde{q}')X}]  \leq \mathbb{E}[X^2] \mathbb{E}[{\rm e}^{-p'(1-\tilde{q}')X}],$$
which is  always the case by Chebychev's other inequality \cite[p.168]{Hardy1952}, because $x^2$ is increasing for $x >0$ and ${\rm e}^{-p'(1-\tilde{q}')x }$ is decreasing for $x>0$. Note that the Poisson distribution is a special case of a Mixed-Poisson distribution.

\subsection{Polarized partial immunity}
\label{subsec:polarBM}

In this subsection we investigate the polarized partial immunity model of Section~\ref{sec:poldef} when $p'=p$, so apart from some individuals being immune in the second epidemic the second disease is the same as the first.   First, in Section~\ref{subsubsec:critalpha}, we study how the reproduction number $R_0^{(2)}$ for the second epidemic  depends on the probability $\alpha$ that an individual infected by the first epidemic is susceptible to the second epidemic, and determine the critical value of $\alpha$ so that there can be no large second epidemic.  Then, in Section~\ref{subsubsection:examplenetworks}, we use two concrete examples of networks to explore the accuracy of the approximation of $R_0^{(2)}$ derived by \cite{Bans12}.

\subsubsection{Critical value of $\alpha$}
\label{subsubsec:critalpha}
\replaced{Consider}{We now study} $R_0^{(2)}$ \added{as given in Theorem \ref{thmpol}} and \deleted{we} use the notation $R_0^{(2)}(\alpha)$ \added{ and $M(\alpha)$ for the corresponding matrix $M$} to stress the dependence of this reproduction number \added{and matrix }on $\alpha$.  We have trivially that $R_0^{(2)}(1)= R_0^{(1)}$. On the other hand, $R_0^{(2)}(0)<1$, because if $\alpha =0$, the second epidemic spreads among the vertices which were susceptible after the first epidemic and the first epidemic underwent a large outbreak by assumption, so it must be subcritical during its final stage \citep{newman2005threshold}.
Furthermore, from the definition of $R_0^{(2)}(\alpha)$ it is easy to deduce that $R_0^{(2)}(\alpha)$ is a continuous\added{ strictly} increasing function of $\alpha$ and therefore
$$\alpha_{c}=\inf \{\alpha\in[0,1] : R^{(2)}_{0}(\alpha)>1\}$$
is the unique solution of $R^{(2)}_{0}(\alpha)=1$, which takes its value in $(0,1)$. Recall that we assume that $R_0^{(1)} >1$.

\begin{lemma}\label{lemalphacrit}
For polarized partial immunity, $\alpha_c$ is given by the smallest solution in $[0,1]$ of the quadratic equation
\begin{multline}\label{eq:3.1}
p^{3}(1-p\mu_{BB})(\mu_{AA}\mu_{BB}-\mu_{AB}\mu_{BA})\alpha^{2}\\
-p\left[(1-p\mu_{BB})(\mu_{AA}+p\mu_{BB})+p(1-p)\mu_{AB}\mu_{BA}\right]\alpha+1-p\mu_{BB}=0.
\end{multline}
\end{lemma}
\begin{proof}
In order to find $\alpha_c$, we note that for polarized partial immunity \added{(with $p'=p$)} at least one eigenvalue of \added{$M(\alpha)$ in }\eqref{eq:8} is 1, if and only if $\alpha$ satisfies~\eqref{eq:3.1}.
The quadratic equation \eqref{eq:3.1} has at least one root in $[0,1]$, since $\alpha_c\in (0,1)$. If it has two roots in $[0,1]$, $\alpha_1 <\alpha_2$ say, then  $R^{(2)}_0(\alpha_1) \geq 1$\added{ (because $R^{(2)}_0(\alpha_1)$ is the largest eigenvalue of $M(\alpha_1)$, while 1 is an eigenvalue of $M(\alpha_1)$)}
and $R^{(2)}_0(\alpha_2)> R^{(2)}_0(\alpha_1)$, as $R^{(2)}_0(\alpha)$ is increasing in $\alpha$.
Thus, $R^{(2)}_0(\alpha_1)=1$, since otherwise $R^{(2)}_0(\alpha) \neq 1$ for all $\alpha \in [0,1]$, contradicting $\alpha_c \in (0,1)$. Hence, $\alpha_c$ is the smallest root of \eqref{eq:3.1} in $[0,1]$.
\end{proof}

\subsubsection{Example networks}
\label{subsubsection:examplenetworks}
In the first example we consider the regular random graph in which all vertices have degree 3 (i.e.\ $\mathbb{P}(D=3)=1$). In the second example
we assume that for some given $k \in \{2,3,\cdots\}$, $\mathbb{P}(\tilde{D}=1) = \mathbb{P}(\tilde{D}=k)=1/2$.  So half of the total degree can be attributed to individuals of degree 1, and the other half to vertices of degree $k$. Furthermore, we assume in this second model that there is no recovery from infection, i.e.\ $p=1$. For both models we compare
the \replaced{the asymptotically exact threshold parameter $R_0^{(2)}$ given in Theorem~\ref{thmpol} and the corresponding $\alpha_c$}{exact asymptotic results from Section \ref{sec:polfinal}} with the approximations derived by \cite{Bans12} (see Appendix \ref{App1}). We note that our second example was chosen to highlight that these approximations may differ considerably from the exact results.

In our first example we assume $p\in (1/2,1]$, in order to guarantee that $R_0^{(1)}>1$ (see \eqref{firstR0}).
 By using $\mathbb{P}(D=3)=1$ in \eqref{eq:4}-\eqref{eq:7} we obtain
 $\mu_{AA}=2(1-\frac{(1-p)^{2}}{p})$,
 $\mu_{AB}=\frac{2(1-p)^{2}}{p}$,
 $\mu_{BA}=\frac{2(2p-1)}{p}$ and
 $\mu_{BB}=\frac{2(1-p)}{p}$.
 Substituting these values into \eqref{eq:3.1} and recalling $p\in (1/2,1]$ yields $\alpha_c$ satisfies $g(\alpha_c)=0$, where
\[
g(x)= [4p^{3}(1-p)x^{2}-2(1-2p+4p^2-2p^3)x+1].
\]
Note that $g(x)$ has a minimum, $g(0)=1>0$ and $g(1) =-[p^2(1-p)^2 +(2p-1)^2]<0$.
So, in this example, the epidemic threshold $\alpha_{c}$ is the unique solution in [0,1] (and the smallest solution on $[0,\infty)$)  of $g(\alpha_c)=0$, which is given by
$$\alpha_{c}=\frac{(1-2p+4p^{2}-2p^{3})-\sqrt{(1-2p+4p^{2}-2p^{3})^{2}-4p^{3}(1-p)}}{4p^{3}(1-p)}.$$
We depict $\alpha_c$ together with the corresponding approximating value from \cite{Bans12} as a function of $p$ in Figure \ref{fig:3}.

In Figure \ref{fig:3} we also plot $R_0^{(2)}(2/3)$ and $R_0^{\text{BM}}(2/3)$ (which is the approximation  from \cite{Bans12} as deduced in \eqref{BMR0} in the Appendix)   as functions of $p$. Note that the model by \cite{Bans12} overestimates $\alpha_{c}$ for this regular network.
We also see that, at least on this specific network, the approximation of \cite{Bans12} is worst for small $p$, but overall offers a reasonable approximation for these parameter values.

\begin{figure}[t]
\centering
\resizebox{0.49\textwidth}{!}{\includegraphics{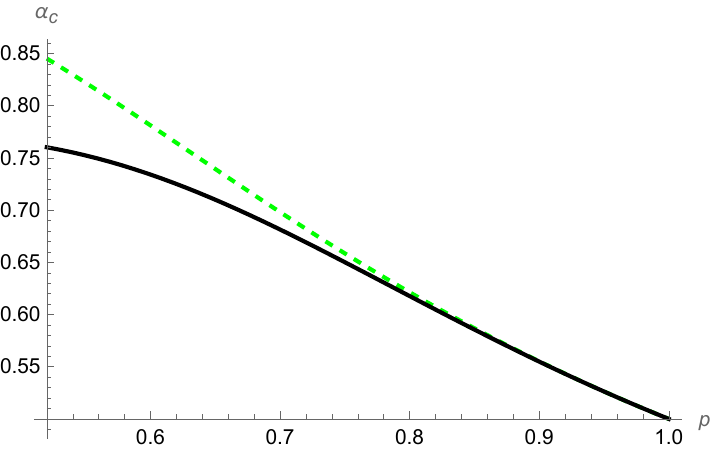}}
\resizebox{0.49\textwidth}{!}{\includegraphics{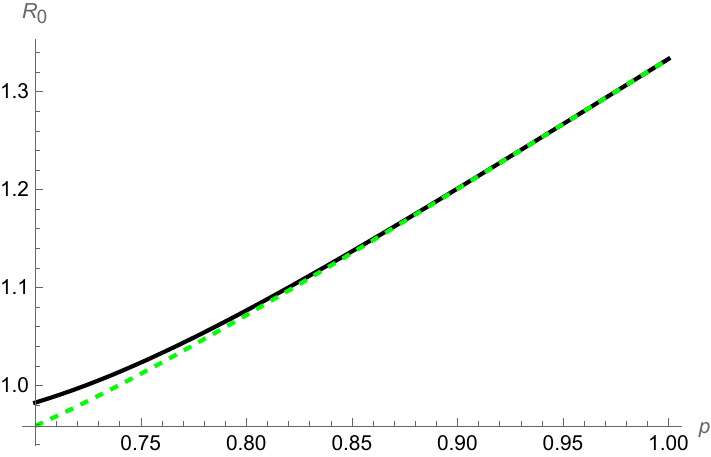}}
\caption{The critical value $\alpha_{c}$ (left) and the basic reproduction number (right) for polarized partial immunity  as a function of the transmission probability $p$.
In both figures the exact values
as deduced in this paper (solid black line) and the approximation of this quantity as given by  \cite{Bans12} (green dashed line) are given for  $\mathbb{P}(D=3)=1$. In the right figure $\alpha=2/3$.}
\label{fig:3}
\end{figure}

In our second example, the constructed graph contains many pairs of neighbours, which have no other neighbours, while if $k$ is large most of the other vertices are in one giant component. Note that pairs of neighbours of degree 1 play no role in the spread of a large epidemic, whether it is the first or second epidemic.

Also for large  $k$, $\tilde{q} \approx 1/2$ and $R_0^{(2)}(\alpha) \approx (k-1) \alpha/2$.
In Figure \ref{fig:5} we compare $R_0^{(2)}(\alpha)$ with $R_0^{\text{BM}}(\alpha)$.  We see in this example (which is constructed exactly for this purpose), that  the difference between the exact  $R_0^{(2)}(\alpha)$ and the approximation $R_0^{\text{BM}}(\alpha)$ may be considerable. In particular, the approximation of \cite{Bans12} yields $\alpha_c \approx 0.38$, while in fact the true value is approximately $2/(k-1) =0.2$.

The intuition behind the approximation failing in this particular case, is that \cite{Bans12} implicitly assume that the network is ``rewired'' among just the ``survivors'' from the first epidemic before the second epidemic starts. Since vertices of degree 1  always end paths of infection, the rewiring means that some vertices of degree 1 that were neighbours of each other before rewiring now block long paths in the network on which the second epidemic spreads, which leads to a lower reproduction number for the second epidemic.

\begin{figure}[t]
\centering
\includegraphics[width=0.5\textwidth]{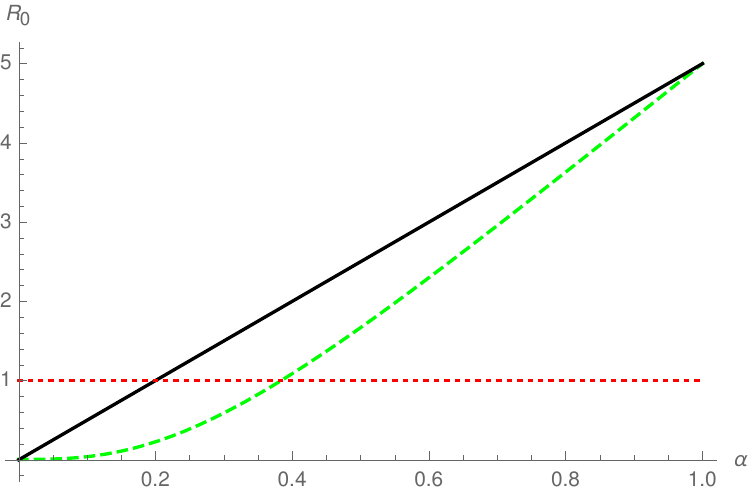}
\caption{The basic reproduction number for polarized partial immunity as a function of $\alpha$, the fraction of individuals infected in the first epidemic that is susceptible again for the second epidemic,  as deduced in this paper (solid black line) and the approximation of this quantity as given in  \cite{Bans12} (green dashed line). Here $\mathbb{P}(\tilde{D}=1)=\mathbb{P}(\tilde{D}=11)=1/2$ and $p=1$.}
\label{fig:5}       
\end{figure}

\subsection{Leaky partial immunity}
 \label{leaky:example}

In this subsection we consider similar examples to those in Section~\ref{subsec:polarBM}, using the model for leaky partial immunity obtained by setting $p'=p$ in~\eqref{pi-leaky1}.  Using~\eqref{eq:8}, the offspring mean matrix for the second epidemic is
\begin{equation}
\label{equ:Mleaky}
M_{\rm leaky}=p\begin{pmatrix}
\alpha_S \alpha_I\mu_{AA} & p \alpha_S \alpha_I \mu_{AB} &(1-p)\alpha_I\mu_{AB} \\
\alpha_S \alpha_I \mu_{BA} & p \alpha_S \alpha_I \mu_{BB} & (1-p)\alpha_I\mu_{BB} \\
\alpha_S\mu_{BA} & 0 & \mu_{BB}
\end{pmatrix}.
\end{equation}
The corresponding mean matrix under the model for polarized partial immunity defined in Section~\ref{sec:poldef}, with $p'=p$, is
\begin{equation}
\label{equ:Mpolarized}
M_{\rm pol}=p\begin{pmatrix}
\alpha \mu_{AA} & p \alpha \mu_{AB} &(1-p)\mu_{AB} \\
\alpha \mu_{BA} & p \alpha_S \mu_{BB} & (1-p) \mu_{BB} \\
\alpha \mu_{BA} & 0 & \mu_{BB}
\end{pmatrix} .
\end{equation}
Suppose that $\alpha=\alpha_S \alpha_I$ and $\alpha_I>0$, and let $\Delta$ be the $3 \times 3$ diagonal matrix with successive diagonal elements $1, 1, \alpha_I^{-1}$.  Then $M_{\rm leaky}=\Delta  M_{\rm pol} \Delta^{-1}$, so $M_{\rm leaky}$ and $M_{\rm pol}$ have the same eigenvalues.  It follows that the basic reproduction number $R_0^{(2)}$ for the second epidemic is the same under both the leaky and polarized models for partial immunity, and hence so is the critical value $\alpha_c$.

For computations, \cite{Bans12} again use an approximation in which they ``reconfigure'' the edges in the configuration model between the first and second epidemic. This reconfiguring is performed in such a way that whether or not a vertex is infected in the first epidemic, as well as its number of infected neighbours in the first epidemic are not changed. \cite{Bans12} analyse the second epidemic (on the reconfigured network) using a two-type bond percolation model.  Although this approach yields an analysis that is exact in the large population limit for an epidemic on the reconfigured network, it does not lead to an asymptotically exact analysis of the second epidemic as the reconfigured network does not yield an asymptotically correct model for the true network following the first epidemic.  The approximating two-type bond percolation model is described in Appendix \ref{App2}, where
the corresponding value of $\alpha$, which we denote by $\alpha_c^{BM}$, so that the second epidemic is critical is derived.   As shown in Apppendix \ref{App2}, the joint-degree distributions given in~\cite{Bans12} that underpin their analysis of the two-type bond percolation model are incorrect. The correct distributions are used in the examples below. This second approximation in \cite{Bans12} can also be used for polarized partial immunity and under it $R_0^{(2)}$, and hence also $\alpha_c^{BM}$, is the same for both leaky and polarized partial immunity when $\alpha_S \alpha_I =\alpha$ (see Appendix \ref{App2}).

\begin{figure}[t]
\begin{center}
\resizebox{0.49\textwidth}{!}{\includegraphics{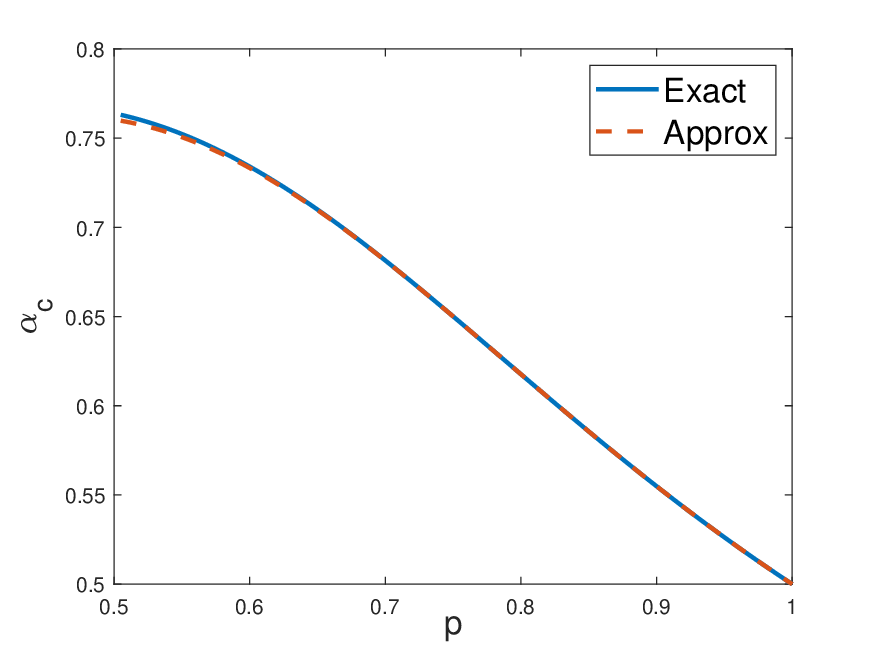}}
\resizebox{0.49\textwidth}{!}{\includegraphics{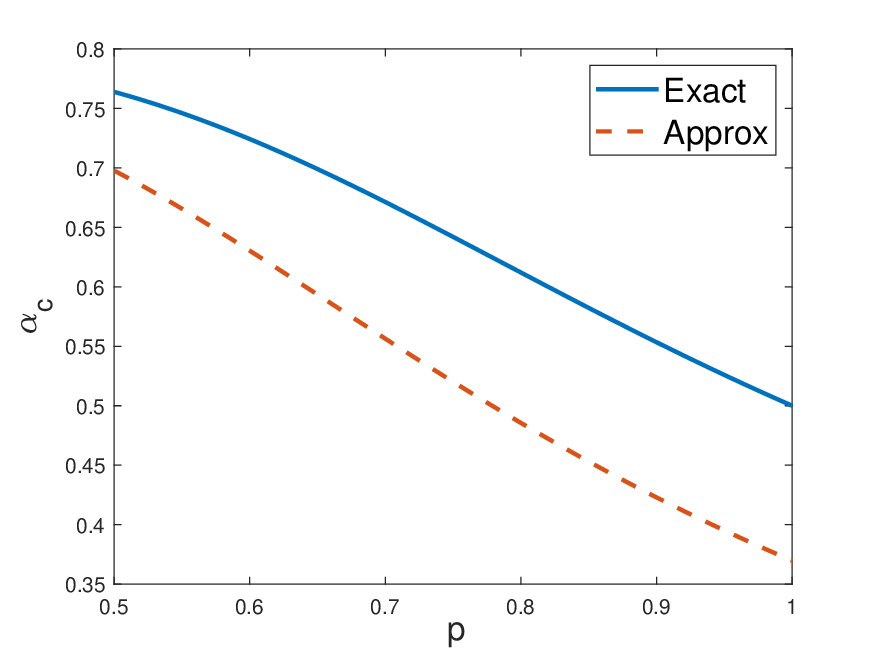}}
\end{center}
\caption{The critical value $\alpha_c$ as a function of the transmission probability $p$ as deduced in this paper (solid line) and the approximation of this quantity using the two-type percolation approximation (dashed line). Left panel: $\mathbb{P}(D=3)=1$; right panel: $\mathbb{P}(\tilde{D}=1)=\mathbb{P}(\tilde{D}=5)=1/2$.  See text for further details.}
\label{fig:alphaccomp}
\end{figure}

We depict $\alpha_c$ and $\alpha_c^{\rm BM}$ as functions of $p$ for the degree distributions given by (i) $\mathbb{P}(D=3)=1$ and (ii)
$\mathbb{P}(\tilde{D}=1)=\mathbb{P}(\tilde{D}=5)=1/2$ in Figure~\ref{fig:alphaccomp}.  Note that both of these degree distributions have $\mathbb{E}[\tilde{D}-1]=2$ and hence $R_0^{(1)}=2p$.  For the first degree distribution, the approximation $\alpha_c^{\rm BM}$ is extremely good, indeed for polarized partial immunity it is clearly superior to that used in Figure~\ref{fig:3}.  However, the approximation performs poorly for the second degree distribution.  In particular, when $p=1$, the approximation $\alpha_c^{\rm BM}$ is appreciably smaller than the correct $\alpha_c$, which has the following intuitive explanation.  When $p=1$, after the first epidemic there is no edge between an infected and a susceptible individual and in the event of a large outbreak all individuals of the giant component of the underlying network are infected.  Thus the reconfiguration of edges between the first and second epidemic underlying the approximation simply involves breaking all edges between individuals infected by the first epidemic into half-edges and repairing them uniformly at random, and doing the same for edges between individuals not infected by the first epidemic.  The degree distribution of individuals in the giant component is stochastically larger than $D$, so the approximation leads to an inflated $R_0^{(2)}$, and hence a reduced $\alpha_c$.

\section{Numerical illustrations of our model}
\label{sec:illus}

In this section, we use simulations to explore the accuracy for finite networks of our asymptotic results for the probability ($p_{\rm maj}$) and relative final size ($z_{\rm maj}$) of a large second epidemic, and investigate briefly the behaviour of $z_{\rm maj}$ on the type (polarized or leaky) of partial immunity and other model parameters, especially the choice of degree distribution.  The examples are based on  \cite[Section 3]{Bans12}.  We consider four degree distributions:
\begin{enumerate}
\item[(i)]
$D \equiv d$, i.e.~$D$ is constant with $p_d=1$  and $f_D(s)=s^d$;
\item[(ii)]
$D \sim {\rm Poisson}(\mu_D)$, i.e.~$D$ is Poisson with mean $\mu_D$  and $f_D(s)={\rm e}^{-\mu_D(1-s)}$;
\item[(iii)]
$D \sim {\rm Geom}(\theta)$, i.e.~$p_k=(1-\theta)^{k-1} \theta$ $(k=1,2,\dots)$, with mean $\mu_D=\theta^{-1}$  and $f_D(s)=\theta s/(1-(1-\theta)s$);
\item[(iv)]
$D \sim {\rm Power}(\alpha_D,\kappa)$, i.e.~$p_k=ck^{-\alpha_D}{\rm e}^{-\frac{k}{\kappa}}$ $(k=1,2,\dots)$, where $\alpha_D, \kappa \in (0,\infty)$
and the normalising constant $c={\rm Li}_{\alpha_D}({\rm e}^{-\frac{1}{\kappa}})$, with ${\rm Li}_{\alpha_D}$ being the polylogarithm function.
\end{enumerate}
\begin{remark}
The fourth distribution is a power law with exponential cut-off (see, for example,~\cite{newman2002spread}) that has been used extensively in the physics literature. Note that, with $\theta={\rm e}^{-\frac{1}{\kappa}}$ and $\beta={\rm Li}_{\alpha_D}(\theta)$, $f_D(s)=\beta^{-1}{\rm Li}_{\alpha_D}(\theta s)$, $f_D^{(1)}(s)=\frac{1}{\beta s}{\rm Li}_{\alpha_D-1}(\theta s)$ and $f_D^{(2)}(s)=\frac{1}{\beta s^2}[{\rm Li}_{\alpha_D-2}(\theta s)-
{\rm Li}_{\alpha_D-1}(\theta s)]$, which enable $R_0^{(1)}, R_0^{(2)}, p_{\rm maj}$ and $z$ to be calculated.
\cite{Bans12} use the scale-free distribution $p_k=k^{-\gamma}/\zeta(\gamma)$ $(k=1,2,\dots)$, where $\zeta$ is the Riemann zeta function and $\gamma$ $(\approx 2.0654)$ is chosen so that $\mu_D=10$.  However, that distribution has infinite variance, which violates the assumptions in Section~\ref{sec:conmod}.
\end{remark}

\subsection{Relating asymptotic theory to outcomes in finite populations}

Each simulation for a population of size $n$ consists of first simulating $D_1, D_2, \dots, D_n$ from the given degree distribution, yielding the numbers of half-edges attached to the individuals in the population, and then pairing these half-edges uniformly at random to form the underlying network, as described in Section~\ref{sec:conmod}.  Two successive epidemics are then simulated on the network so constructed.  For the first epidemic, the initial infective is chosen uniformly at random  from the population.  For the second epidemic, under polarized partial immunity, each individual infected in the first epidemic is rendered immune independently with probability $1-\alpha$, and the initial infective is then chosen uniformly at random from all individuals that remain susceptible.  Under leaky partial immunity, an individual, $i_*$ say, is chosen uniformly at random from the whole population.  If $i_*$ was not infected by the first epidemic then it is the initial infective for the second epidemic.  If $i_*$ was infected by the first epidemic then
with probability $\alpha_S$ it is the initial infective for the second epidemic, otherwise a new potential initial infective $i_*'$ is chosen independently according to the same law, and the process is repeated until an initial infective is obtained.

Let $p^I_1$ and $p^I_2$ be the individual-to-individual infection probabilities for the two epidemics, assuming no immunity in the second epidemic.
We set $p^I_1=0.15, p^I_2=0.3$ and consider networks with $\mu_D=10$, as in \cite[Fig.~5]{Bans12}.  The values of $\sigma_D^2$ under the four distributions are (i) 0, (ii) 10, (iii) 90 and (iv) 271.4950, yielding values of $R_0^{(1)}$, calculated from~\eqref{firstR0}, of $1.35, 1.5, 2.7$ and $5.4224$, respectively.  We use the \replaced{formulation of leaky partial immunity given by~\eqref{pi-leaky1} with $p'=p^I_2$}
{same formulation of leaky partial immunity as in Section~\ref{leaky:example}, and set $\pi_{00}=p^I_2$, $\pi_{01}=\alpha_S p^I_2$, $\pi_{10}=\alpha_I p^I_2$ and $\pi_{11}=\alpha_S\alpha_I p^I_2$}, so $R_0^{(2)}$ is the same under the two models of partial immunity provided $\alpha_S \alpha_I=\alpha$. In particular, the two  $R_0^{(2)}$s are equal if $(\alpha_S, \alpha_I)=(\sqrt{\alpha}, \sqrt{\alpha)}$
and that is used, with $\alpha=\frac{1}{2}$, in Table~\ref{table:sims}, which gives estimates of the probability, $p_{\rm maj}^{(n)}$, and relative final size, $z_{\rm maj}^{(n)}$, of a large second epidemic, conditional upon a large first epidemic, for a network of size $n$, obtained from simulations as follows.


For each combination of population size $n$, degree distribution and partial immunity type, $10,000$ simulations of a network and two successive epidemics on it were performed.  The estimate $\hat{p}_{\rm maj}^{(n)}$ of $p_{\rm maj}^{(n)}$ is given by the fraction of those simulations that had a large first epidemic that also had a large second epidemic, while the estimate $\hat{z}_{\rm maj}^{(n)}$ of $z_{\rm maj}$ is given by the sample mean of the fractions of the population infected in the latter second epidemics. Standard errors of the estimates are given in brackets in Table~\ref{table:sims}.  An epidemic in a population having $n=500$ was deemed large if its size exceeded $100$.  If $n=1,000$ the cut-off was $200$, whilst for $n=5,000$  and $n=10,000$ the cut-off was $300$.  These cut-offs were chosen by inspection of corresponding histograms.  For the parameter values used in the simulations there is a clear distinction between small and large epidemics.  The asymptotic values $p_{\rm maj}$ and $z_{\rm maj}$, calculated using \replaced{Theorems~\ref{thmsurvival}-\ref{thmpol} in Section~\ref{sec:Results}}{the methods described in Section~\ref{sec:problarge}} are
given in the $n=\infty$ rows of the table.  The method of choosing the initial infective for the second epidemic implies that the probability that infective is of type 3 is $q/(q+\alpha(1-q))$ for polarized partial immunity and $q/(q+\alpha_S(1-q))$ for leaky partial immunity.

\begin{table}
\centering
\begin{tabular}{|c|c|c|c|c|c|}
\hline
& & \multicolumn{2}{|c|}{\mbox{Polarized}} & \multicolumn{2}{|c|}{\mbox{Leaky}}
\\
$D$ &$n$ & $\hat{p}_{\rm maj}^{(n)}$ &  $\hat{z}_{\rm maj}^{(n)}$ & $\hat{p}_{\rm maj}^{(n)}$ &  $\hat{z}_{\rm maj}^{(n)}$ \\
\hline
 & 500 &  0.8705 \;(0.0046) &    0.6369  \;(0.0009)& 0.8382 \;(0.0050) &   0.8389 \;(0.0006) \\
 & 1000 &  0.8720 \;(0.0046) &    0.6333 \;(0.0006)& 0.8457 \;(0.0049) &   0.8393 \;(0.0004)\\
$D \equiv 10$  & 5000 & 0.8680 \;(0.0046) & 0.6313 \;(0.0003)& 0.8403 \;(0.0049) &    0.8399 \;(0.0002) \\
 & 10000 &  0.8704  \;(0.0045) &    0.6315 \;(0.0002) & 0.8359 \;(0.0050)  &  0.8400 \;(0.0001) \\
 & $\infty$ &  0.8720  & 0.6311 & 0.8508  & 0.8397  \\
\hline
 & 500 &  0.8018 \; (0.0053) &   0.5764 \; (0.0008) & 0.7969 \; (0.0053) &    0.7984 \; (0.0005) \\
 & 1000 & 0.8040 \;(0.0052)  &  0.5743 \;(0.0005) & 0.7935 \; (0.0053) &   0.7984 \;  (0.0004) \\
$\mbox{Poisson}(10)$  & 5000 & 0.8111 \; (0.0051) &   0.5741 \; (0.0002) & 0.7973 \; (0.0053) &   0.7984 \;   (0.0002)  \\
& 10000 &  0.7990 \;(0.0052)  &  0.5740 \; (0.0002) & 0.7930 \; (0.0053) &   0.7986  \;  (0.0001) \\
& $\infty$ &  0.8098  &  0.5738 &  0.8053 &  0.7986 \\
\hline
 & 500 & 0.5591 \;(0.0066) &   0.4127 \;(0.0007)& 0.6437 \;(0.0064) &    0.6481 \;(0.0005)  \\
 & 1000 & 0.5738 \;(0.0066) &   0.4114 \;(0.0005)& 0.6475 \;(0.0064) &    0.6472 \;(0.0003) \\
$\mbox{Geom}(1/10)$  & 5000 &  0.5686 \;(0.0066) &    0.4110 \;(0.0002) & 0.6553 \;(0.0064) &  0.6469 \;(0.0001)\\
  & 10000 &  0.5725 \;(0.0066) &    0.4106 \;(0.0002) & 0.6438 \;(0.0064)  &  0.6467 \;(0.0001) \\
  & $\infty$ &  0.5706 &0.4106& 0.6294 & 0.6468  \\
\hline
 & 500 &  0.4242 \;(0.0073) &    0.3294 \;(0.0008) & 0.5258 \;(0.0074) &  0.5310\;(0.0005) \\
 & 1000 & 0.4242 \;(0.0073) &    0.3292 \;(0.0006) & 0.5216 \;(0.0073) &  0.5309 \;(0.0004)\\
$\mbox{Power}(1,36.6472)$  & 5000 & 0.4225 \;(0.0073) & 0.3275 \;(0.0003) & 0.5313 \;(0.0073) & 0.5305 \;(0.0002)  \\
  & 10000 &  0.4443  \;(0.0073) &    0.3274\;(0.0002) & 0.5221 \;(0.0074) & 0.5305 \;(0.0001) \\
  & $\infty$ &  0.4273 & 0.3274 & 0.5037  & 0.5304  \\
\hline
\end{tabular}
\caption{{ Simulation results against theoretical (asymptotic) calculations for probability and relative final size of a large second epidemic, given a large first epidemic. See text for further details.}}
\label{table:sims}
\end{table}

Table~\ref{table:sims} indicates that for these parameter values the asymptotic approximations of both $p_{\rm maj}$ and $z_{\rm maj}$ are good, even for $n=500$, with the
approximation of $\hat{z}_{\rm maj}^{(n)}$ being generally better than that of $\hat{p}_{\rm maj}^{(n)}$.  Note that under this calibration, polarized partial immunity gives better protection in the sense that the size of a large second epidemic, if one occurs, is smaller.  However, for the constant and Poisson degree distributions, the probability of a large second epidemic is slightly larger.  Note also that for degree distributions in Table~\ref{table:sims}, the probability and mean size of a large second epidemic both decrease as $\sigma_D^2$ increases.  This is because the first epidemic is greater for larger $\sigma_D^2$ (cf.~the values of $R_0^{(1)}$ given above) and hence the population is better protected against the second epidemic.

\subsection{Exploring model behaviour}

Figure~\ref{fig:zsecond} shows plots of the fraction of the population infected by a large second epidemic, $z_{\rm maj}$, against the individual-to-individual infection probability for the second epidemic, $p_2^I$, for various models of partial immunity when the individual-to-individual infection probability for the first epidemic, $p_1^I=0.15$.  For polarized partial immunity, we set $\alpha=\frac{1}{2}$, as before.  For leaky partial immunity, we consider three models with
$\alpha_S\alpha_I=\frac{1}{2}$ and also the model with $\alpha_S=\alpha_I=\frac{1}{2}$.  The latter uses the calibration considered by \cite{Bans12}, who note that ``the total susceptibility and total transmissibility over the entire network is equal in the two immunity models".
However, for disease transmission these factors tend to operate in a multiplicative rather than additive fashion.
If $\alpha_S=\alpha_I=\alpha$, then every element of $M_{\rm leaky}$ (see~\eqref{equ:Mleaky}) is less than or equal to the corresponding element of $M_{\rm pol}$  (see~\eqref{equ:Mpolarized}), with strict inequality for some elements provided $\alpha \in (0,1)$.  It follows, in an obvious notation, that if $\alpha \in (0,1)$ then $R_{0,\text{leaky}}^{(2)} <R_{0,\text{pol}}^{(2)}$, as reflected in the intercepts of the curves in Figure~\ref{fig:zsecond} with the $p_2^I$-axis.  If instead $\alpha_S=\alpha_I=\sqrt{\alpha}$, then $R_{0,\text{leaky}}^{(2)} = R_{0,\text{pol}}^{(2)}$, as is also reflected by the intercepts of the curves in Figure~\ref{fig:zsecond}.

\begin{figure}[t]
\begin{center}
\resizebox{6cm}{!}{\includegraphics{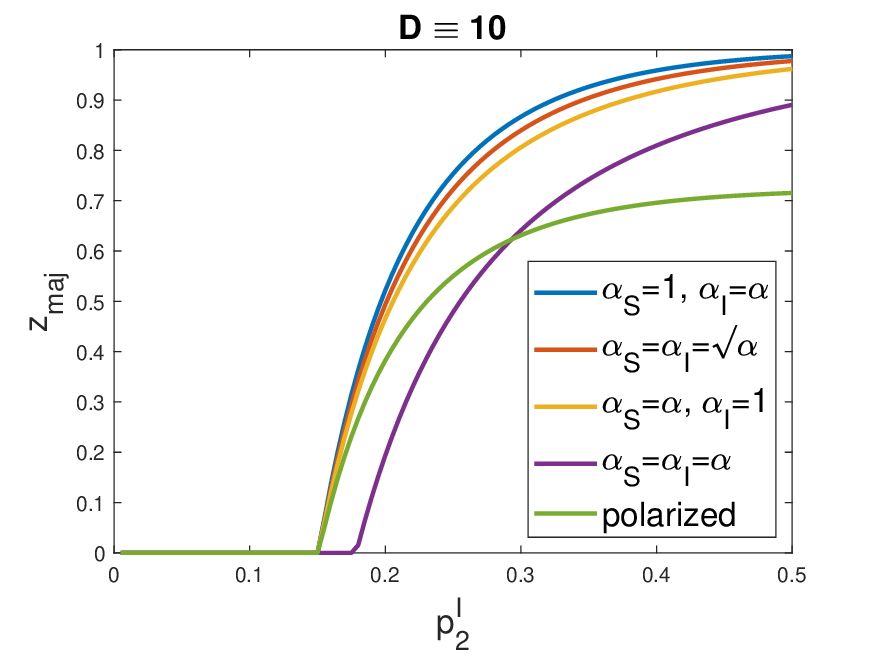}}
\resizebox{6cm}{!}{\includegraphics{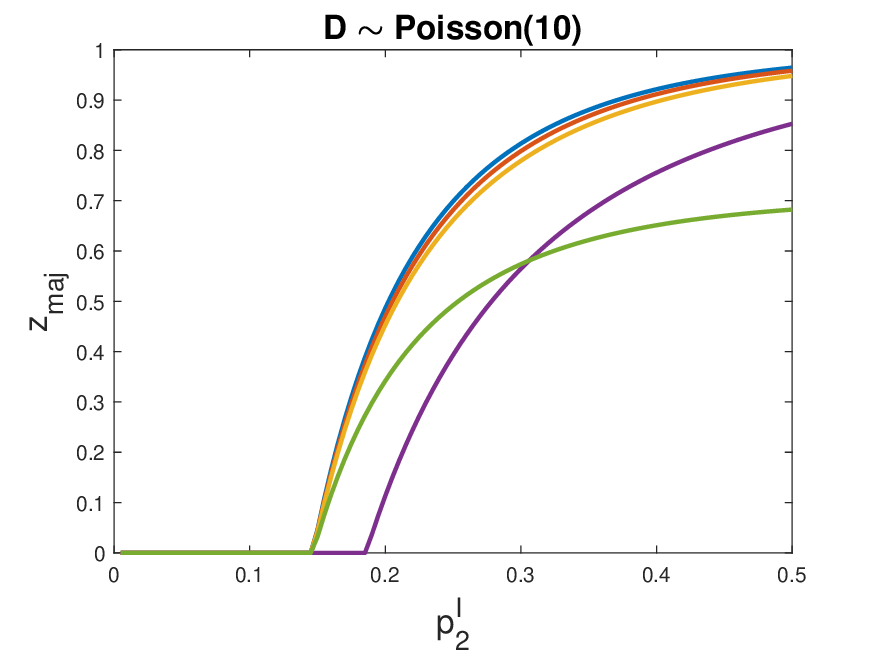}}\\
\resizebox{6cm}{!}{\includegraphics{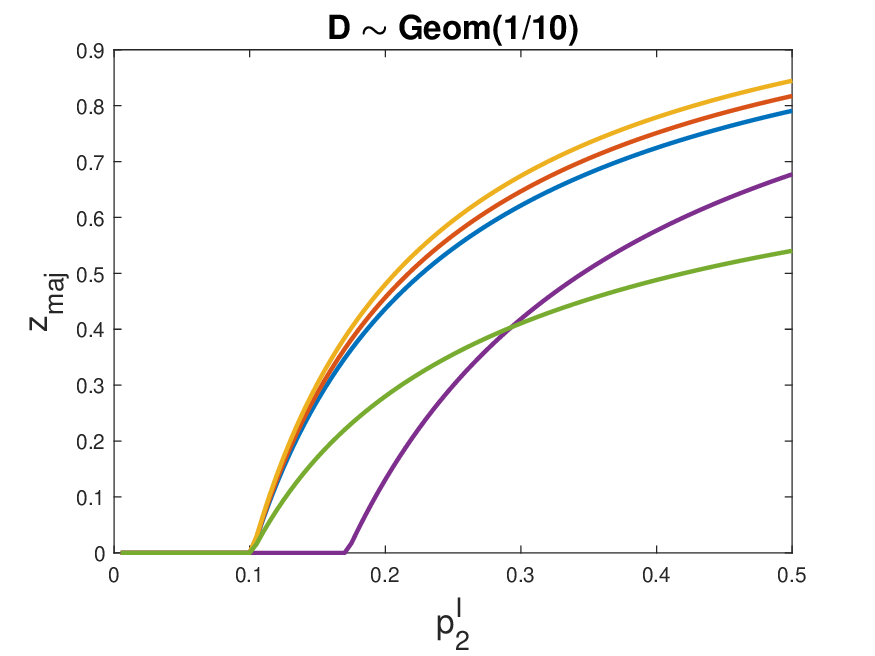}}
\resizebox{6cm}{!}{\includegraphics{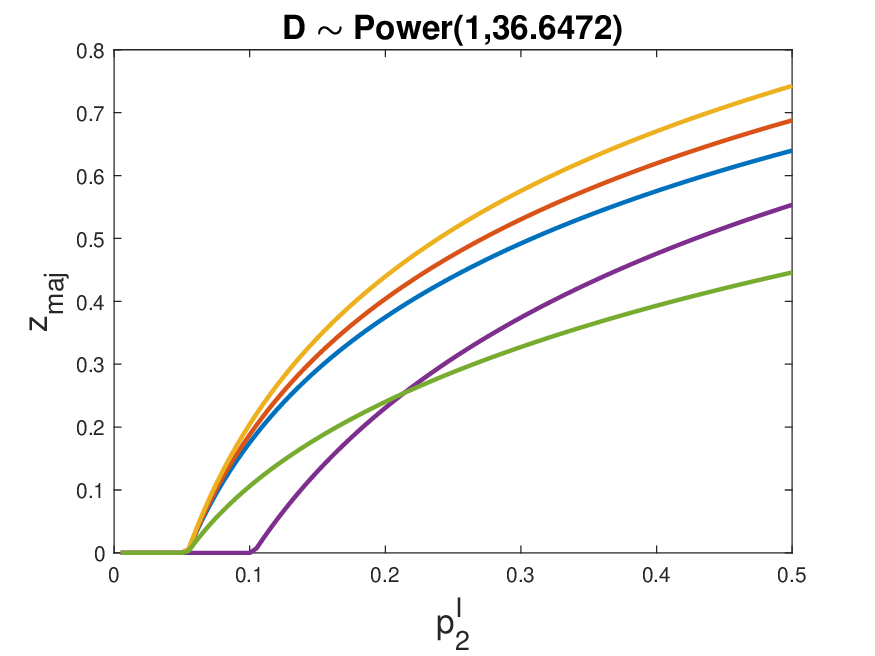}}
\end{center}
\caption{Relative final size $z_{\rm maj}$ of a large second outbreak.  See text for further details.}
\label{fig:zsecond}
\end{figure}

One way of considering the effects of the two models of partial immunity is to imagine that each individual infected by the first epidemic is given two coins, a susceptibility coin that has probability $\alpha_S$ of landing heads and an infectivity coin that has probability $\alpha_I$ of landing heads.  Under the leaky model, every time a susceptible individual is contacted they toss their susceptibility coin and become infected if and only if it lands heads.  If they become infected then they attempt to contact their neighbours independently, each with probability $p_2^I$, but for each such attempt they toss their infectivity coin and the contact is successful if and only if it lands heads.  Under the polarized model, the first time that a susceptible is contacted they toss their susceptibility coin.  If it lands heads then they become infected and then contact their neighbours independently, each with probability $p_2^I$; they do not toss their infectivity coin.  However, if it lands tails then they are immune from that contact and all subsequent contacts.

Suppose that $(\alpha_S, \alpha_I)=(\alpha,1)$.  Then it is easily seen that polarized partial immunity (with parameter $\alpha$) leads to a smaller second epidemic than leaky partial immunity, since an individual's susceptibility coin landing tails protects them also from {\it all} further contacts.  Suppose instead that $(\alpha_S, \alpha_I)=(\alpha,\alpha)$.  Then leaky partial immunity is likely to lead to a smaller second epidemic at low values of $p_2^I$, since individuals are unlikely to be contacted twice in the second epidemic and individuals infected in both the first and second epidemics are less infectious than under polarized partial immunity.  However, at high values of $p_2^I$, individuals are likely to be contacted several times in the second epidemic and the advantage gained from the lifelong immunity in the polarized model more than offsets individuals being more infectious if infected.  This explains the crossover effect of the corresponding curves in Figure~\ref{fig:zsecond}.

\section{Analysis of the second epidemic and proofs}\label{sec:proofs}

\subsection{A percolation approach to the first and second epidemic}\label{percapproach}
\replaced{To}{Before we} prove Theorems \ref{thmR0} and \ref{thmsurvival}\added{,} we use a colouring of the vertices and edges in the graph $\mathcal{G}$ and in a tree $\mathcal{T}$, in which the ``local neighbourhood'' of vertices is distributed as the local neighbourhood of vertices in $\mathcal{G}$\added{ as $n \to \infty$}.

\subsubsection{Colouring the graph $\calG$}
In what follows we exploit that in the first epidemic, for different pairs of neighbours the events that the individual first infected (if any) makes an infectious contact with the other are independent and all have probability $p$. We do this through colouring the vertices and edges in $\calG$. The coloured graph we denote by $\calGst^{(n)} =\calGst = (\calVst,\calEst)$, and we now describe its construction.

Note that each edge in $\calE$ can be used for transmission of the disease at most once during an epidemic, namely, when one of the endpoints is infectious and the other susceptible, since after such a contact both endpoints of the edge are infectious \citep{cox1988limit}. Hence, for every edge in $\calG$, we may determine in advance whether or not it serves as a possible route of transmission of infection, provided one of its end vertices becomes infected.
We create $\calEst$ by colouring the edges in $\calE$ independently, each edge being red with probability $p$ and blue otherwise.
By construction, the subset of red edges in $\calEst$ is distributed as the set of open edges after bond percolation  with edge-probability $p$ on the graph $\calG$ \citep{cox1988limit,Grimmett1999}. Red edges are interpreted as possible routes of transmission of the disease, that is, if there is a red edge present between two vertices, then the first vertex that becomes infectious (if any) infects the second vertex, if the second vertex has not\deleted{ yet} been infected by another vertex in the meantime.

Let $v_0$ denote the initially infectious vertex in the first epidemic. We obtain $\calVst$ by colouring the vertices in $\calV$ such that $v_0$ and all vertices in $\calV$ which can be reached by a red path in $\calEst$ from $v_0$ are coloured red, while all other vertices are blue. The set of red vertices is denoted by  $\calC(v_0)$. By construction, the set of individuals infected during the first epidemic is distributed as $\calC(v_0)$. So, loosely speaking, the red vertices and edges are those that are infected and along which infectious contacts \replaced{would be}{are} made \added{if one of the end-vertices of the edge got infected }during the first epidemic.

Using the obvious connection between $\calGst$ and percolation on $\calG$ we call a subset $\mathcal{S}$ of $\calVst$ a component if $\mathcal{S}$ is connected through red edges in $\calGst$ (i.e.\ every vertex in $\mathcal{S}$ can be reached from all other vertices in $\mathcal{S}$ through a path of red edges in $\calEst$) and there is no red edge connecting $\mathcal{S}$ to its complement.
If $R_0^{(1)}>1$ we say that the epidemic or outbreak is large if $\calC(v_0)$ is the largest component of $\calVst$.  We denote this largest component by $\maxcl$. With probability tending to 1 as $n \to \infty$, $\maxcl$ is the only component  which contains more than $(\log n)^2$ vertices  and there exists a constant \replaced{$q <1$}{$q>0$}, such that the fraction of vertices in the largest component converges in probability to $1-q$ as $n \to \infty$ \cite[Thm.\ 3.2.2]{Durr07}.
It follows immediately that\added{ in the first epidemic} the probability of a large outbreak converges to $1-q$ and conditioned on a large outbreak the fraction of the population infected converges in probability to $1-q$, where $q$ is defined in \eqref{qdef0}.

\subsubsection{A coloured tree approximation of $\calGst$}
\label{sec:coltree}

For the second epidemic we assume that the first outbreak is large and therefore that $\calC(v_0) = \maxcl$.  Thus, we assume implicitly that $R_0^{(1)}>1$.
We consider the coloured graph $\calGst$  and pick a vertex of $\calGst$, say $w_0$,\comment{replaced $u_0$ by $w_0$} uniformly at random. This vertex represents the initial case for the second epidemic. Then $w_0$ may or may not be part of the cluster $\maxcl$.
By construction of the graph $\calG$, the environment of $w_0$ is with high probability locally tree-like, i.e.\ for any $k \in \mathbb{N}$  the number of circuits (a circuit is a path of {\it distinct} edges for which the starting vertex of the first edge is the same as the end vertex of the last edge)  in $\calG$ that both contain $w_0$ and are of length at most $k$  converges in probability to 0 as $n \to \infty$.
We use this locally tree-like property of $\calG$ to describe the environment of $w_0$ first in $\calGst$ and then  in $\calG$. We call $w_0$ the root of the tree.

Let $\calN^{(n)}(w_0;k)$ be the set of (coloured) vertices within distance $k$ of $w_0$ in $\calGst^{(n)}$ together with the (coloured) edges with both end-vertices in this set.
Let $\calT$ be a graph created by a two stage Galton-Watson tree \cite[Chapter 3]{van2016random} in which the ancestor (root) has offspring distribution $D$ and the other particles in the Galton-Watson tree have offspring distribution $\tilde{D}-1$, i.e.\ including their parent they have degree distribution $\tilde{D}$.
The vertex set of $\calT$ corresponds to the particles in the Galton-Watson process and vertices are connected by an edge if and only if they have a parent-child relationship.
Let $\calTst$ be a coloured version of $\calT$, in which edges are independently red with probability $p$, and otherwise blue.
A vertex is coloured red if it is part of an infinite path of red edges in $\calTst$. Formally this means that
a vertex at distance $d$ ($d \in \mathbb{N}$) from the root  is red if and only if it is connected through a path of red edges to at least one vertex at every distance larger than $d$ from the root ($w_0$). Note that if $\calT$ is finite (which occurs with strictly positive probability if $p_0+p_1>0$), then all vertices in $\calTst$ are blue.
Let $\calTst(k)$ be the subgraph of $\calTst$ consisting of all vertices within distance $k$ of the root and the edges of $\calT$ with both end-vertices in this set, while the colour of vertices and edges is preserved.

We use without formal proof that  for all $k \in \mathbb{N}$, $\calN^{(n)}(w_0;k)$ converges in distribution to $\calTst(k)$ as $n \to \infty$.  See \cite{newman2002spread} for a heuristic discussion of the convergence of the non-coloured version of the graphs and \cite{Ball09} for a more formal analysis.
The convergence of the coloured graphs can be understood heuristically by observing that with high probability, i.e.~with probability tending to 1 as $n \to \infty$, the unique large component in $\calGst$ corresponds to the infinite components in $\calTst$. Note that in large but finite graphs, two uniformly chosen vertices that are part of large components are with high probability in the same component, although their distance goes to infinity as $n \to \infty$ \cite[p.~84]{Durr07}, while on a percolated infinite tree, different vertices may be in different infinite components (which can be interpreted as  the vertices being at infinite distance \replaced{from}{of} each other).

We then analyse the second epidemic as if it spreads on $\calTst$ starting with the root initially infectious. We assume that the infectivity and susceptibility of the red vertices behave in the same way as they do in $\calGst$.
We use, again without a formal proof, that the probability of a large second outbreak on $\calGst$ converges to the probability of an infinite outbreak on $\calTst$, and we define and compute a natural threshold parameter $R_0^{(2)}$ for the second epidemic.

In the remainder of the paper we say that for  every non-root vertex $u$ in $\calTst$, the neighbour of $u$ in $\calTst$ which is on the shortest path from the root to $u$  is the \emph{parent} of  $u$ and this parent is denoted by $a(u)$. All other neighbours are \emph{children} of $u$. Similarly, if $u$ is on the shortest path from the root to a vertex $v$, then $v$ is a \emph{descendant} of $u$. Finally, we call an infinite path in $\calTst$ of descendants of $u$ which starts at $u$ a \emph{path of descent of $u$}. We say that a path of descent is red if and only if all edges in the path are red.

For every vertex $v$ in $\calTst$, let $\calTst^{(-v)}$ be $\calTst$ with $v$ and all adjacent edges made blue. The spread of the first epidemic creates three types of vertices: vertices which were infected in the first epidemic and have a red path of descent (type 1), vertices which were infected in the first epidemic but have no red path of descent (type 2),  and vertices which escaped the first epidemic and therefore have no red path of descent (type 3). Note that \deleted{type 2 }vertices \added{of type 2 }are always connected by a red edge to their parent which is either of type 1 or of type 2. Hence, vertex $u$ in $\calTst$ is of
\begin{description}
 \item[type 1,] if $u$ is part of an infinite red path in $\calTst^{(-a(u))}$;
 \item[type 2,] if $u$ is red but not part of an infinite red path in $\calTst^{(-a(u))}$;
 \item[type 3,] if $u$ is blue.
\end{description}
Note also that a vertex $u$ in $\calTst$ is of type 1 if and only if a red path of descent in $\calTst$ starts \replaced{at}{in} $u$. This typing of vertices is necessary to construct a branching process approximation to the second epidemic. Below we see that the joint distributions of the numbers of children of the various types
depends on the type of the individual/vertex whose children are under consideration, and therefore the different types are needed.

As an example, consider $\calG$ with  $p_3 = \mathbb{P}(D=3) =1$. The tree $\calT$ has a root with three children and all other vertices have two children. A realisation of the (randomly) coloured tree $\calTst$ is shown in Figure \ref{fig:tree}. In this figure vertices $u_7$, $u_{17}$ and $u_{19}$ are of type 2. All other red vertices are of type 1, and all blue vertices are of type 3.

\begin{figure}[ht!]

\includegraphics[width=\textwidth]{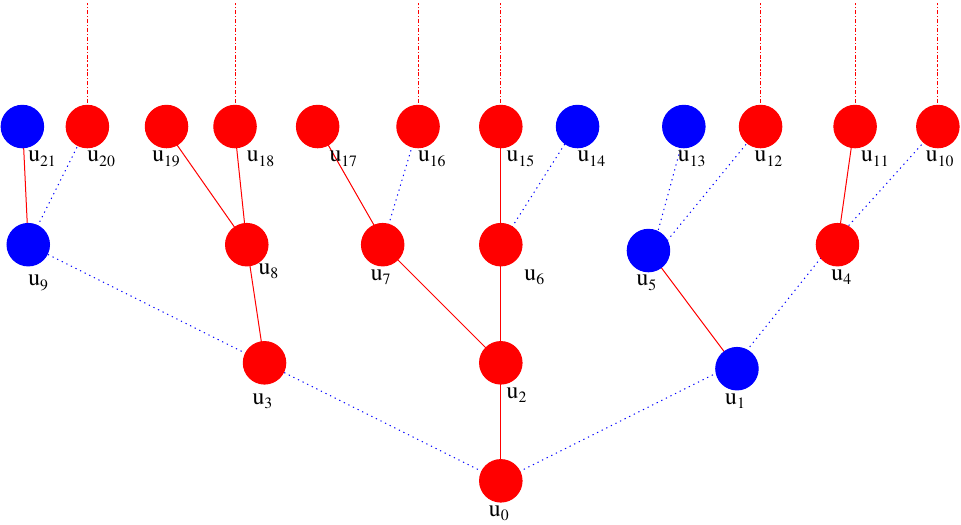}
\caption{The coloured tree $\calTst$ with root $w_0$. In this network every vertex has degree $3$ and each vertex
other than $w_0$ has one ``parent'' and two ``children''. The red vertices represent individuals infected in the first epidemic and blue vertices represent uninfected individuals. The solid red edges represent relationships through which the first epidemic would transmit if at least one of the end-vertices were infected in the first epidemic. The dotted blue edges represent relationships through which the first epidemic would not transmit, even if one of the end vertices got infected. The dashed-dotted red lines represent (infinite) red paths of descent in $\calTst$.}
\label{fig:tree}       
\end{figure}

The typing of the vertices in  $\calTst$ allows us to approximate the second epidemic on $\calTst$ by a multi-type Galton-Watson process \citep[Chapter 4]{jagers1975branching} in which
the ancestor has a different offspring distribution than
all other individuals in the process.  We denote this branching process by $\mathcal{B}_F^{(2)}$.  The superfix ${(2)}$ indicates that the approximation is for the second epidemic.  The suffix $F$ indicates that the branching process approximates the forward process of the epidemic, in contrast to the branching process $\mathcal{B}_B^{(2)}$
introduced in Section~\ref{sec:finalsize}, which approximates the backward process (or susceptibility set \citep{Ball19}).

\subsection{Branching process approximation for the second epidemic}
\label{branproc}

For $i,j \in \{1,2,3\}$, let $\tilde{X}_{ij}$ be the random number of type-$j$ vertices infected by a type-$i$ vertex in $\calTst$, which is not the root.
Define $m_{i j}= E[\tilde{X}_{ij}]$ and  $M$ as the $3 \times 3$ matrix with  $m_{ij}$ as the $i,j$ element. That is, $M$ is the mean offspring matrix for the non-initial generations of $\mathcal{B}_F^{(2)}$.\added{ As becomes clear this matrix is the same as $M$ defined in \eqref{eq:8}.} The dominant eigenvalue of the matrix $M$ is a threshold parameter for the second epidemic and is denoted by $R^{(2)}_{0}$. A major
outbreak of the second epidemic occurs asymptotically with strictly positive probability if and only if $R^{(2)}_{0}>1$.

To make further progress,
label type-1 vertices by $A$,  and type-$2$ and type-$3$ vertices by $B$. The  vertices with label $A$ are the vertices with a red path of descent in $\calTst$, while the vertices with label $B$ do not have a red path of descent.
The number of children of a typical (\added{i.e.} non-root) vertex in $\mathcal{T}$ is distributed as $\bar{D}-1$ and a vertex has label $B$, if all of its children either have label $B$ or are connected\added{ to that vertex} by a blue edge. Therefore, the probability that a vertex in the tree has label $B$ satisfies \eqref{qdef} and standard theory from branching processes gives that it is the smallest positive solution of this equation. That is,  the probability that a typical vertex in the tree has label $B$ is $\tilde{q}$, while the probability that a vertex in the tree has label $A$ is $1-\tilde{q}$.
We note that by Bayes' Theorem, for integers $0\leq \ell \leq k-1$, the probability that a typical vertex (say $v$) with label $A$ has $k-1$ children in total (i.e.\ $v$ is of degree $k$), of which $\ell$ are of type $A$, satisfies
\begin{align}
\label{OffsAA1}
\mathbb{P}&(\text{
$v$ has $k-1$ children of which $\ell$ have label $A$ $\mid$ $v$ has label $A$ })\nonumber\\
&\qquad=\mathbb{P}(\text{$v$ has $k-1$ children of which $\ell$ have label $A$})\nonumber\\
&\qquad \qquad\times
\frac{\mathbb{P}(\text{$v$ has label $A$  $\mid$
$v$ has $k-1$ children of which $\ell$ have  label $A$})}{\mathbb{P}(\text{$v$ has label $A$)}}.
\end{align}

Because $v$ has label $A$  if and only if $v$ shares a red edge with at least one child with label $A$, we have $$\mathbb{P}(\text{$v$ has label $A$  $\mid$
$v$ has $k-1$ children of which $\ell$ have  label $A$}) = 1-(1-p)^{\ell}$$ and  \eqref{OffsAA1} yields
\begin{align}
\label{OffsAA}
\mathbb{P}&(\text{
$v$ has $k-1$ children of which $l$ have label $A$ $\mid$ $v$ has label $A$ })\nonumber\\
&\qquad=
\tilde{p}_k \dbinom{k-1}{\ell}(1-\tilde{q})^{\ell} \tilde{q}^{k-\ell-1} \times \frac{ 1-(1-p)^{\ell}}{1-\tilde{q}}.
\end{align}


We can compute in a similar fashion,
\begin{align}
\label{OffsBA}
\mathbb{P}&(\text{
$v$ has $k-1$ children of which $\ell$ have label $A$ $\mid$ $v$ has label $B$ })\nonumber\\
&\qquad=
\tilde{p}_k \dbinom{k-1}{\ell}(1-\tilde{q})^{\ell} \tilde{q}^{k-\ell-1} \times \frac{ (1-p)^{\ell}}{\tilde{q}}.
\end{align}
Obviously, because all children have either label $A$ or $B$, we can deduce
$$\mathbb{P}(\text{
$v$ has $k-1$ children of which $\ell$ have label $B$ $\mid$ $v$ has label $i$ }),
 $$ for $i \in \{A,B\}$, immediately from \eqref{OffsAA} and \eqref{OffsBA}.
Below we use the above probabilities to compute the threshold parameter $R_0^{(2)}$, and the probability and size of a large outbreak of the second epidemic.

\subsection{Proof of Theorem \ref{thmR0}: Computation of $R_0^{(2)}$}\label{sec:R0}

In order to compute $R_0^{(2)}$, which is the dominant eigenvalue of
the mean offspring matrix $M$, we first need to derive, for $i,j \in \{A,B\}$, the  expected number of label-$j$ children of a given label-$i$ individual, which we denote by $\mu_{ij}$.
Expressions for these expectations are derived as follows, using \eqref{OffsAA} and \eqref{OffsBA}. Let $v$ be a vertex, which is not the root.
By definition we have

\begin{align*}
\mu_{AA}&=\mathbb{E}[\text{
number of label-$A$ children of $v$ $\mid$ $v$ has label $A$ }]\\
&=\frac{ \mathbb{E}\left[(\text{
number of label-$A$ children of $v$})\,\mathds{1}
(\text{$v$ has label $A$})\right]}{\mathbb{P}(\text{$v$ has label $A$})}\\
&=\frac{\displaystyle\sum_{\ell=1}^{\infty}\ell\,\mathbb{P}(\text{$v$ has $\ell$ label-$A$ children})\,\mathbb{P}
(\text{$v$ has label $A$ $\mid$ $v$ has $\ell$ label-A children})}
{\mathbb{P}(\text{$v$ has label $A$})}.
\end{align*}

Conditioning on the degree of $v$ then gives
\begin{align*}
\mu_{AA} &= \frac{1}{1-\tilde{q}} \displaystyle\sum_{k=1}^{\infty}\tilde{p}_{k}\displaystyle\sum_{\ell=0}^{k-1} \ell \dbinom{k-1}{\ell}(1-\tilde{q})^{\ell} \tilde{q}^{k-\ell-1}(1-(1-p)^{\ell})\\
&=\displaystyle\sum_{k=2}^{\infty}(k-1)\tilde{p}_{k}\displaystyle\sum_{\ell=1}^{k-1}\dbinom{k-2}{\ell-1}(1-\tilde{q})^{l-1} \tilde{q}^{k-\ell-1}(1-(1-p)^{\ell}).
\end{align*}
Finally, using the binomial theorem, we obtain \eqref{eq:4}:
\begin{equation*}
\mu_{AA} =\displaystyle\sum_{k=2}^{\infty}(k-1)\tilde{p}_{k}\left(1-(1-p)(1-p+p\tilde{q})^{k-2}\right).
\end{equation*}
Note that $\mu_{AA}$ may be written as
$$
\frac{d}{dx} \left.\displaystyle\sum_{k=2}^{\infty}\tilde{p}_{k}x^{k-1}\right|_{x=1}
-
(1-p) \frac{d}{dx} \left.\displaystyle\sum_{k=2}^{\infty}\tilde{p}_{k} x^{k-1}\right|_{x=1-p+p\tilde{q}}.
$$
That is,
$$\mu_{AA} = f_{\tilde{D}-1}'(1) -(1-p)f_{\tilde{D}-1}'(1-p+p \tilde{q}) = \frac{1}{\mu_D} \left( f_D''(1) -(1-p)f_D''(1-p+p \tilde{q})\right),  $$
where $f_{\tilde{D}-1}(x)=\mathbb{E}[x^{\tilde{D}-1}]$ and $f_D(x)=\mathbb{E}[x^D]$ ($x \in[0,1]$) are the probability generating functions (PGFs) for the random variables $\tilde{D}-1$ and $D$, respectively. This representation is useful, because for many probability distributions the PGFs have a convenient form.
Similarly, we obtain \eqref{eq:5},
\begin{align*}
\mu_{AB}&=\frac{1}{1-\tilde{q}} \displaystyle\sum_{k=1}^{\infty}\tilde{p}_{k}\displaystyle\sum_{\ell=0}^{k-1}\dbinom{k-1}{\ell}(1-\tilde{q})^{\ell} \tilde{q}^{k-\ell-1}(k-\ell-1)(1-(1-p)^{\ell})\\
&=\frac{\tilde{q}}{1-\tilde{q}} \displaystyle\sum_{k=2}^{\infty}(k-1)\tilde{p}_{k}\displaystyle\sum_{\ell=0}^{k-2}\dbinom{k-2}{\ell}(1-\tilde{q})^{\ell} \tilde{q}^{k-\ell-2}(1-(1-p)^{\ell})\\
&=\frac{\tilde{q}}{1-\tilde{q}}\displaystyle \sum_{k=2}^{\infty}(k-1)\tilde{p}_{k}\left(1-(1-p+p\tilde{q})^{k-2}\right),
\end{align*}
so $$\mu_{AB}= \frac{1}{\mu_D} \left( f_D''(1) -f_D''(1-p+p \tilde{q})\right).$$
With similar computations we obtain \eqref{eq:6} and \eqref{eq:7}, where we can write
\begin{align*}
\mu_{BA} & = \frac{(1-p)(1-\tilde{q})}{\tilde{q}}\displaystyle\sum_{k=2}^{\infty}(k-1)\tilde{p}_{k}(1-p+p\tilde{q})^{k-2}\\
& = \frac{(1-p)(1-\tilde{q})}{\tilde{q}} \frac{1}{\mu_D} f_D''(1-p+p \tilde{q})
\end{align*}
and
\begin{equation*}
\mu_{BB}= \displaystyle\sum_{k=2}^{\infty}(k-1)\tilde{p}_{k}(1-p+p\tilde{q})^{k-2}=  \frac{1}{\mu_D} f_D''(1-p+p \tilde{q}).
\end{equation*}
So we obtain the matrix $M$ from \eqref{eq:8}
\begin{equation*}
M=\begin{pmatrix}
\pi_{11}\mu_{AA} & p\,\pi_{11}\mu_{AB} &(1-p)\pi_{10}\mu_{AB} \\
\pi_{11}\mu_{BA} & p\,\pi_{11}\mu_{BB} & (1-p)\pi_{10}\mu_{BB} \\
\pi_{01}\mu_{BA} & 0 & \pi_{00}\mu_{BB}
\end{pmatrix}.
\end{equation*}
Note that a label-$B$ child of a red individual is type 2 if and only if it is infected by its parent in the first epidemic, and all label-$B$ children of a blue individual are necessarily type 3.
The threshold parameter $R^{(2)}_0$ is given by the dominant eigenvalue of the matrix $M$ \citep[Theorem 7.3]{diekmann2012mathematical} and Theorem \ref{thmR0} follows \added{(see also \cite[Chapter 4]{jagers1975branching})}.

\subsection{Proof of Theorem \ref{thmsurvival}: The probability of a large outbreak}
\label{sec:problarge}
In order to prove Theorem  \ref{thmsurvival} and compute the probability of a large outbreak of the second epidemic, we need to compute the
PGFs of the joint offspring distributions of $\mathcal{B}_F^{(2)}$.
We need to consider the offspring random variables for the root of the branching process (as defined at the start of Section \ref{branproc}), and the non-initial generations of the approximating branching process separately.

Note that the root can only be of type 1 or type 3.  For $i \in \{1,3\}$ let $\mathbf{Y}_i=(Y_{i1},Y_{i2},Y_{i3})$ be a random vector where $Y_{ij}$ ($j \in \{1,2,3\}$) is the number of type-$j$ children of the root of $\calTst$ if this root is of type $i$ and let $\mathbf{X}_i=(X_{i1},X_{i2},X_{i3})$ denote a random vector, where $X_{ij}$ ($j \in \{1,2,3\}$) is the number of type-$j$ children infected in the second epidemic by the root of $\calTst$ if this root is of type $i$.
To understand the difference between $Y_{ij}$ and $X_{ij}$ note  that whether or not vertex $v$ is a child of the root does not depend on the spread of the second epidemic, but only on whether or not $v$ is a neighbour of the root of $\calTst$.

For $\mathbf{s}=(s_{1},s_{2},s_{3})\in [0,1]^3$,
define the PGFs
\begin{eqnarray*}
f_1(\mathbf{s}) & = & \mathbb{E}[s_{1}^{X_{11} }s_{2}^{X_{12} } s_{3}^{X_{13}}],\\
f_3(\mathbf{s}) & = & \mathbb{E}[s_{1}^{X_{31}} s_{3}^{X_{33} }],\\
 f_{\mathbf{Y}_1}(\mathbf{s}) &= & \mathbb{E}[s_{1}^{Y_{11} }s_{2}^{Y_{12} } s_{3}^{Y_{13}}],\\  f_{\mathbf{Y}_3}(\mathbf{s}) &= & \mathbb{E}[s_{1}^{Y_{31} }s_{2}^{Y_{32} } s_{3}^{Y_{33}}].
\end{eqnarray*}
We show that these definitions are consistent with \eqref{fifirst}.
Observe $f_3(\mathbf{s}) = \mathbb{E}[s_{1}^{X_{31}} s_{2}^{X_{32}} s_{3}^{X_{33} }]$, because $X_{32}=0$ by definition.
To compute $f_{i}(\mathbf{s})$ ($i\in \{1,3\}$) we use that
\begin{equation}
\label{equ:PGFXtilde}
f_{i}(\mathbf{s})=\mathbb{E}[s_{1}^{X_{i1} }s_{2}^{X_{i2} } s_{3}^{X_{i3} }] = \mathbb{E}\left[\mathbb{E}\left[s_{1}^{X_{i1} }s_{2}^{X_{i2} } s_{3}^{X_{i3} }\mid Y_{i1},Y_{i2},Y_{i3}\right]\right].
\end{equation}
Because infections in the second epidemic are independent, we have
\begin{equation}\label{XYtrans}
\mathbb{E}[s_{1}^{X_{11} }s_{2}^{X_{12} } s_{3}^{X_{13} }|Y_{11},Y_{12},Y_{13}]]
=(1-\pi_{11}+\pi_{11}s_1)^{Y_{11}}(1-\pi_{11}+\pi_{11}s_2)^{Y_{12}}(1-\pi_{10}+\pi_{10}s_3)^{Y_{13}}.
\end{equation}
Similarly,
\begin{equation}\label{XYtrans2b}
\mathbb{E}[s_{1}^{X_{31} }s_{2}^{X_{32} } s_{3}^{X_{33} }|Y_{31},Y_{32},Y_{33}]]
=  (1-\pi_{01}+\pi_{01}s_1)^{Y_{31}}(1-\pi_{00}+\pi_{00}s_3)^{Y_{33}},
\end{equation}
where we have used that $X_{32} =0$.

Recall from \eqref{tdefs} that
\begin{eqnarray*}
t_{ij}(s_j)& = & 1-\pi_{11}+\pi_{11}s_j,\\
t_{i3}(s_3) &= & 1-\pi_{10}+\pi_{10}s_3,\\
t_{31}(s_1) &= & 1-\pi_{01}+\pi_{01}s_1,\\
t_{32}(s_2)& =& 0,\\
t_{33}(s_{3}) & = & 1-\pi_{00}+\pi_{00}s_3.
\end{eqnarray*}
It immedialtely follows  from~\eqref{equ:PGFXtilde}-\eqref{XYtrans2b} that, for $i\in\{1,3\}$ and $\mathbf{s}=(s_{1},s_{2},s_{3})\in [0,1]^3$,
$$f_{i}(\mathbf{s}) =  f_{\mathbf{Y}_i}(t_{i1}(s_1),t_{i2}(s_2),t_{i3}(s_3)).$$
To complete the derivation of  $f_{i}(\mathbf{s})$ $(i \in \{1,3\})$ we determine
$f_{\mathbf{Y}_i}(\mathbf{s})$ $(i \in \{1,3\})$.
First,
\begin{align}
\label{eq:1a}
f_{\mathbf{Y}_1}(\mathbf{s})&=\frac{1}{1-q}\displaystyle\sum_{k=1}^{\infty} p_{k}\displaystyle\sum_{\ell_{1}=0}^{k}\dbinom{k}{\ell_{1}}(1-\tilde{q})^{\ell_{1}} \tilde{q}^{k-\ell_{1}}(1-(1-p)^{\ell_1})\nonumber\\
&\qquad \qquad \qquad\times  \displaystyle\sum_{\ell_{2}=0}^{k-\ell_1}\dbinom{k-\ell_1}{\ell_{2}}p^{\ell_{2}} (1-p)^{k-\ell_{1}-\ell_2}
s_1^{\ell_1}s_2^{\ell_2}s_3^{k-\ell_1-\ell_2}\nonumber\\
&=\frac{1}{1-q}\displaystyle\sum_{k=1}^{\infty}p_{k}\big(\left(1-\tilde{q})s_1+\tilde{q}(p s_2+(1-p)s_3\right)\big)^{k}\nonumber\\
&\qquad -\frac{1}{1-q}\displaystyle\sum_{k=1}^{\infty}p_{k}\big(\left(1-p)(1-\tilde{q})s_1+\tilde{q}(ps_2+ (1-p)s_3\right)\big)^{k}\nonumber\\
&=\frac{1}{1-q}f_{D}(\left(1-\tilde{q})s_1+\tilde{q}(p s_2+(1-p)s_3\right))\nonumber\\
&\qquad -\frac{1}{1-q}f_{D}(\left(1-p)(1-\tilde{q})s_1+\tilde{q}(p s_2+(1-p)s_3\right)).
\end{align}
The first equality can be understood by noting that if a vertex ($v$ say) of degree $k$ has an (infinite) red path of descent at least one of $v$'s children has to have a red path of descent, and if there are $\ell_1$ children with a red path of descent, at least one of them has to be connected through a red edge to $v$, hence the $1-(1-p)^{\ell_1}$ term. The remaining $k-\ell_1$ children  do not have a red path of descent and among those $k-\ell_1$, $\mathrm{Bin}(k-\ell_{1},p)$ are of type 2 (cf.~the derivation of~\eqref{OffsAA}). The second equality follows using the binomial theorem.

In the same way we deduce
\begin{equation}\label{eq:1.4a}
f_{\mathbf{Y}_3}(\mathbf{s})
= \frac{1}{q}\displaystyle f_{D}\left((1-p)(1-\tilde{q})s_1+\tilde{q}s_3)\right).
\end{equation}

Next we consider the offspring in $\mathcal{B}^{(2)}_F$ of individuals other than the root.
For $i,j \in \{1,2,3\}$,  let $\mathbf{\tilde{Y}}_i=(\tilde{Y}_{i1},\tilde{Y}_{i2},\tilde{Y}_{i3})$ denote a random vector where $\tilde{Y}_{ij}$ is the number of type-$j$ children of a type-$i$ individual in  $\calTst$ that is not the root. Moreover, for $i,j \in \{1,2,3\}$, let $\mathbf{\tilde{X}}_i=(\tilde{X}_{i1},\tilde{X}_{i2},\tilde{X}_{i3})$ be a random vector,
where $\tilde{X}_{ij}$ is the random number of type-$j$ vertices infected in the second epidemic by a type-$i$ vertex in $\calTst$, which is not the root.
For $\mathbf{s}=(s_{1},s_{2},s_{3})\in [0,1]^3$ and $i \in \{1,2,3\}$, let
$f_{\mathbf{\tilde{Y}}_i}(\mathbf{s})=\mathbb{E}[s_{1}^{\tilde{Y}_{i1} }s_{2}^{\tilde{Y}_{i2} } s_{3}^{\tilde{Y}_{i3} }]$.

Using the same arguments as above we obtain that
\begin{equation}
\label{equ:PGFXtilde1a}
\tilde{f}_{i}(\mathbf{s})=f_{\mathbf{\tilde{Y}}_i}(t_{i1}(s_1),t_{i2}(s_2),t_{i3}(s_3))\quad(i \in \{1,2,3\}, \mathbf{s} \in [0,1]^3).
\end{equation}
Furthermore, by similar calculations to the derivations of~\eqref{eq:1} and~\eqref{eq:1.4}, we obtain \eqref{eq:2a}-\eqref{eq:2c}.

The extinction probability (and hence the probability of a minor outbreak) is now computed using the theory of multi-type branching processes \cite[Chapter 4]{jagers1975branching}.
Assume that $R_0^{(2)}>1$.

Let $\bxi=(\xi_{1},\xi_{2},\xi_{3})$ be the unique solution of the system of equations $\xi_1 = \tilde{f}_1(\xi_1,\xi_2,\xi_3)$, $\xi_2 = \tilde{f}_2(\xi_1,\xi_2,\xi_3)$ and $\xi_3 = \tilde{f}_3(\xi_1,\xi_2,\xi_3)$ in $[0,1)^{3}$. Then for $i \in \{1,2,3\}$, $\xi_i$ is the probability that the (second epidemic) offspring of an infected type-$i$ individual (other than the root) goes extinct.
Thus,
the probability of a minor outbreak  if the root of $\calGst$ is of type $i$ ($i \in \{1,3\}$)
is given by
$f_{i}(\bxi)$.
The unconditional probabilities of minor and large outbreaks depend on how the root is chosen. If the root is a vertex chosen uniformly from the population, then the probability of a large outbreak is $1-q f_{1}(\bxi)-(1-q)  f_{3}(\bxi)$. This finishes the proof of Theorem \ref{thmsurvival}.

We note that it might be more natural to let the probability that a vertex is the root of the second epidemic depend on the susceptibility of this vertex, in which case the computation of the probability of the root is dependent on the specific model under consideration (see Section~\ref{sec:illus}).

\subsection{Proof of Theorem \ref{thmfinalsize}: Final size of a large outbreak}
\label{sec:finalsize}

A key tool in determining the fraction of the population that is infected by a large outbreak, in the limit as the population size $n \to \infty$ is the concept of \emph{susceptibility sets} (see e.g.~\cite{Ball02,Ball19,Brit19}). To introduce this concept, consider the first epidemic and relax the assumption that the infectious period random variable $L$ is almost surely constant.

Construct a random directed graph, which we denote by $\mathcal{G}_1$, on the vertex set $\mathcal{V}$ as follows.  For each $i \in \mathcal{V}$, determine who $i$ would contact if they were to become infected by first sampling its infectious period $L_i$ independently from $L$ and then, conditional upon $L_i$, drawing a directed edge from $i$ to each of its neighbours in the graph independently with probability $1-{\rm e}^{-\beta L_i}$.

For distinct $i$ and $j$ in $\mathcal{V}$, write $i \leadsto j$ if and only if there is a chain of directed edges
from $i$ to $j$.  For $i \in \mathcal{V}$, define the susceptibility set of $i$ by $\mathcal{S}_i^{(1)}=\{j \in \mathcal{V} \setminus \{i\}\,:  j \leadsto i\}$.  Note that an initially susceptible individual, $i$ say, is infected by the first epidemic if and only if $\mathcal{S}_i^{(1)}$ has non-empty intersection with the set of initial infectious individuals, hence the terminology.

Suppose that $n$ is large and that the epidemic is started by a single initial infectious individual chosen uniformly at random
from $\mathcal{V}$.  The size $S_i^{(1)}$ of the susceptibility set $\mathcal{S}_i^{(1)}$ can be approximated by the total progeny
of a \textit{backward} Galton-Watson branching process, $\mathcal{B}_B^{(1)}$ say, defined as follows. The branching process  $\mathcal{B}_B^{(1)}$ has one ancestor, who corresponds to individual $i$ and is excluded from the above total progeny.  Note that $i$ has degree
distributed according to $D$ and each of $i$'s neighbours would infect $i$ independently with probability $p=1-\mathbb{E}[{\rm e}^{-\beta L}]$
if they were to become infected.  Thus the offspring random variable for the initial generation of $\mathcal{B}_B^{(1)}$
has the mixed-binomial distribution ${\rm Bin}(D,p)$, that is a binomial distribution with random ``size parameter'' $D$ and ``probability parameter'' $p$.  Now consider a typical neighbour, $j$ say, of $i$ that joins the susceptibility set $\mathcal{S}_i$.  The number of neighbours of $j$ is distributed according to
$\tilde{D}$ but one of those neighbours is $i$, so the offspring random variable for this and all subsequent
generations of $\mathcal{B}_B^{(1)}$  has the  mixed-binomial distribution ${\rm Bin}(\tilde{D}-1,p)$.  The generation-based approximation
of the susceptibility set $\mathcal{S}_i^{(1)}$ by the branching process $\mathcal{B}_B^{(1)}$ continues in the obvious fashion.

It is intuitively plausible that for large $n$, individual $i$ is infected by a large outbreak if and only if
the branching process $\mathcal{B}_B^{(1)}$ that approximates $S_i$ does not go extinct.  Further, let $Z_n^{(1)}$ denote the size of the first epidemic.  Then, subject to mild regularity conditions, conditional upon the occurrence of a large outbreak (more formally one that infects at least $\log n$ individuals), $n^{-1}Z_n^{(1)}$ converges in probability
to $1-q$ as $n \to \infty$, where $q$ is the extinction probability of $\mathcal{B}_B^{(1)}$; see~\cite{ball_sirl_trapman_2014} for a proof when the
underlying graph $\mathcal{G}$ is a random intersection graph.  It is easily checked that $q$ is given by~\eqref{qdef0}.

\added{The approach with susceptibility sets is close to a probability generating functions approach to calculate the size of the giant in-component in directed or semi-directed configuration model networks as described in \citep{Kenah07,Mill07,Kenah11}. The susceptibility set of a node is then identical to the in-component of the node in the corresponding epidemic percolation network. For the first epidemic, neither the calculations using susceptibility sets nor those using epidemic percolation networks become conceptually harder, when not assuming fixed infectious periods \citep{Ball02,Kenah07b,Ball09}. However, as discussed in Section \ref{sec:disc}, the analysis of the second epidemic will become harder if one does not assume a fixed infectious period for the first epidemic.}

We now assume that \added{for the first epidemic the infectious period is fixed and has length 1: }$\mathbb{P}(L=1)=1$ and consider the second epidemic, assuming that the first epidemic resulted in a large outbreak.  Let $\mathcal{G}_2$ denote the random directed graph corresponding to $\mathcal{G}_1$ above
but for the second epidemic; i.e.~in which there is a directed edge from $i$ to $j$ if and only if individual
$i$, if infected by the second epidemic, contacts individual $j$ in that epidemic.  As in Section~\ref{sec:coltree},
pick a root $w_0$ uniformly at random from all vertices in $\mathcal{V}$, independently from the root chosen in Section~\ref{sec:coltree}.  The susceptibility set, $\mathcal{S}_{w_0}^{(2)}$ say, of individual $w_0$ for the second epidemic is derived from $\mathcal{G}_2$ in the obvious fashion.  The susceptibility set $\mathcal{S}_{w_0}^{(2)}$
can be approximated using a $3$-type Galton-Watson process, $\mathcal{B}_B^{(2)}$ say, with individuals typed as in the ``forward process'' of infection analysed in
Section~\ref{sec:coltree}, but with the new choice of $w_0$. As in Section~\ref{sec:problarge}, the root $w_0$ is necessarily of type $1$ or $3$.

For $i,j \in \{1,2,3\}$, let $\hat{X}_{ij}$ be the number of type-$j$ offspring of a typical type-$i$
individual in a non-initial generation of the branching process $\mathcal{B}_B^{(2)}$ and define
$\check{X}_{ij}$ $(i \in \{1,3\}, j \in \{1,2,3\})$ similarly for the initial generation.  Further, for
$\mathbf{s} \in [0,1]^3$, let
\[
\begin{array}{rl}
\hat{f}_{i}(\mathbf{s}) & =\mathbb{E}[s_{1}^{\hat{X}_{i1} }s_{2}^{\hat{X}_{i2} } s_{3}^{\hat{X}_{i3} }]
\quad (i \in \{1,2,3\}),\\
\check{f}_{i}(\mathbf{s}) & =\mathbb{E}[s_{1}^{\check{X}_{i1} }s_{2}^{\check{X}_{i2} } s_{3}^{\check{X}_{i3} }]
\quad (i \in \{1,3\}).
\end{array}
\]
For $i \in \{1,3\}$, let $\check{\xi}_i$ be the extinction probability of $\mathcal{B}_B^{(2)}$ given that there is one ancestor, whose type is $i$.  Omitting the details, it is easily verified that if $R_0^{(2)}>1$ then
$\check{\xi}_1<1$ and $\check{\xi}_3<1$, otherwise $\check{\xi}_1=\check{\xi}_3=1$. By standard multi-type
branching process theory (cf.~the end of Section~\ref{sec:problarge}), if $R_0^{(2)}>1$ then $\check{\xi}_i=\check{f}_i(\hat{\bxi})$, $i \in \{1,3\}$, where $\hat{\bxi}=(\hat{\xi}_1, \hat{\xi}_2, \hat{\xi}_3)$
is the unique solution in $[0,1)^{3}$ of the system of equations  $\hat{\xi}_i = \hat{f}_i(\hat{\bxi})$, $i \in \{1,2,3\}$.

Let $Z_n^{(2)}$ be the size of the second epidemic.
Note that, in the limit as $n \to \infty$, given a large first epidemic the root $w_0$ has type $1$
with probability $1-q$ (where $q$ is given by~\eqref{qdef0}) otherwise it has type $3$. Then, given that the first and second epidemics are both large, we conjecture that $n^{-1}Z_n^{(2)}$ converges in probability to $z=1-(1-q)\check{\xi}_1-q\check{\xi}_3$
as $n \to \infty$.  (This conjecture is supported by simulations, see Section~\ref{sec:illus}, and also by equivalent results for other models; for example
\cite[Thm 6.2]{Ball09}).
 \deleted{Conditional upon both epidemics being large, as $n \to \infty$, the probability an individual chosen uniformly at random from the population avoids both epidemics converges to $q\check{\xi}_3$, the probability it is infected by the first epidemic but not the second converges to $(1-q)\check{\xi}_1$, the probability it is infected by the second epidemic but not the first converges to $q(1-\check{\xi}_3)$ and the probability it is infected by both epidemics converges to $(1-q)(1-\check{\xi}_1)$.}

To compute $\check{\xi}_1$ and $\check{\xi}_3$, we need expressions for the PGFs $\hat{f}_{i}(\mathbf{s})$ $(i \in \{1,2,3\})$
and $\check{f}_{i}(\mathbf{s})$ $(i \in \{1,3\})$.  We use that
\begin{equation}
\label{equ:tower}
\hat{f}_{i}(\mathbf{s})
=\mathbb{E}[\mathbb{E}[s_{1}^{\hat{X}_{i1} }s_{2}^{\hat{X}_{i2} } s_{3}^{\hat{X}_{i3} }|\tilde{Y}_{i1},\tilde{Y}_{i2},\tilde{Y}_{i3}]] \quad (i \in \{1,2,3\})
\end{equation}
and
\begin{equation}
\label{equ:tower1}
\check{f}_{i}(\mathbf{s})
=\mathbb{E}[\mathbb{E}[s_{1}^{\check{X}_{i1} }s_{2}^{\check{X}_{i2} } s_{3}^{\check{X}_{i3} }|
Y_{i1},Y_{i2},Y_{i3}]] \quad (i \in \{1,3\}),
\end{equation}
where $\tilde{Y}_{ij}$ $(i,j \in  \{1,2,3\})$ and $Y_{ij}$ $(i \in \{1,3\}, j \in \{1,2,3\})$ are as defined in Section~\ref{sec:problarge}.

\added{
The calculations of $\mathbb{E}[s_{1}^{\hat{X}_{i1} }s_{2}^{\hat{X}_{i2} } s_{3}^{\hat{X}_{i3} }|\tilde{Y}_{i1},\tilde{Y}_{i2},\tilde{Y}_{i3}]$ for $i \in \{1,2,3\}$ and\\ $\mathbb{E}[s_{1}^{\check{X}_{i1} }s_{2}^{\check{X}_{i2} } s_{3}^{\check{X}_{i3} }|Y_{i1},Y_{i2},Y_{i3}]$ for $i \in \{1,3\}$ are essentially the same as those for the corresponding conditional PGFs used in determining the probability of a large outbreak in Section~\ref{sec:problarge},
except that $\pi_{ij}$ $(i,j \in \{0,1\})$ defined at~\eqref{pidef} are replaced by  $\hat{\pi}_{ij}$
$(i,j \in \{0,1\})$, where $\hat{\pi}_{ij}=\pi_{ji}$ $(i,j \in \{0,1\})$. This is because the branching process
$\mathcal{B}_B^{(2)}$ approximates the backward process associated with the second epidemic, whereas $\mathcal{B}_F^{(2)}$ approximates the corresponding forward process.  It follows using~\eqref{equ:tower} that\comment{PT Oct11: took off tilde}
\[
\hat{f}_{i}(\mathbf{s})=f_{\tilde{\mathbf{Y}}_i}(\hat{t}_{i1}(s_1), \hat{t}_{i2}(s_2), \hat{t}_{i3}(s_3))\quad(i \in \{1,2,3\}, \mathbf{s} \in [0,1]^3),
\]
where $\hat{t}_{ij}(s_j)$ $(i,j \in \{1,2,3\})$ \replaced{are given by~\eqref{tdefs2}}{is the same as in Theorem \ref{thmfinalsize}}.
Similarly,
\[
\check{f}_1(\mathbf{s})=f_{\mathbf{Y}_i}(\hat{t}_{i1}(s_1),\hat{t}_{i2}(s_2),\hat{t}_{i3}(s_3))\quad(i \in \{1,3\}, \mathbf{s} \in [0,1]^3).
\]
\replaced{The expressions for $\hat{f}_{i}(\mathbf{s})$ $(i \in \{1,2,3\})$ and $\check{f}_1(\mathbf{s})$ $(i \in \{1,3\})$ coincide with those given at~\eqref{equa:PGFXtilde} and~\eqref{equa:PGFXcheck}, so Theorem~\ref{thmfinalsize} follows.}{This concludes the proof.}
}\comment{PT Oct 11, took away some text}

\subsection{Proof of Theorem \ref{thmpol}: Polarized partial immunity}
\label{sec:polproof}
We analyse the model with polarized partial immunity by considering the spread of the second epidemic on the coloured tree $\calTst$ (see Section~\ref{sec:coltree}) in a similar fashion as is done above for the model of Section~\ref{subsec:secondepi}.  For the analysis of the probability of a large outbreak of an epidemic on $\calTst$ (or on any tree) it is important that an individual can only be infected by its parent. So, whether or not a vertex is infected by its parent in $\calTst$ is independent of other infection events in the tree. Therefore,
we may reduce the  probabilities of transmission to all nodes infected in the first epidemic by the factor $\alpha$ and obtain $\pi_{00}= p', \pi_{01}=\alpha p', \pi_{10}=p'$ and $\pi_{11}=\alpha p'$, as at~\eqref{pi-pol}.  The proofs of the results corresponding to Theorems~\ref{thmR0} and~\ref{thmsurvival} then follow as above.

Turning to the final size of a large outbreak in the second epidemic, we derive now the conditional expectations in equations~\eqref{equ:tower} and~\eqref{equ:tower1} for the offspring PGFs of the backward branching process $\mathcal{B}_B^{(2)}$.
Consider first a vertex, $u$ say, in the coloured tree $\calTst$, which is not the root. For $i \in \{1,2,3\}$,  if $u$ is of type $i$ then the numbers of children of $u$ of the three types in $\calTst$  is distributed
as $(\tilde{Y}_{i1},\tilde{Y}_{i2},\tilde{Y}_{i3})$.
Suppose that $u$ is of type $1$ or $2$, i.e.\ that $u$ was infected by the first epidemic. With probability $1-\alpha$, $u$ is immune to the second epidemic and $\hat{X}_{i1}= \hat{X}_{i2}= \hat{X}_{i3}=0$.
Otherwise each of $u$'s children in $\calTst$  contacts $u$ in the second epidemic independently with probability $p'$.  Hence,
for $i \in \{1,2\}$,
\[
\mathbb{E}[s_{1}^{\hat{X}_{i1} }s_{2}^{\hat{X}_{i2} } s_{3}^{\hat{X}_{i3} }|\tilde{Y}_{i1},\tilde{Y}_{i2},\tilde{Y}_{i3}]=1-\alpha+\alpha(1-p'+p's_1)^{\tilde{Y}_{i1}}(1-p'+p's_2)^{\tilde{Y}_{i2}}(1-p'+p's_3)^{\tilde{Y}_{i3}}.
\]
Thus, using~\eqref{equ:tower},
\[
\hat{f}_{i}(\mathbf{s})=1-\alpha+\alpha f_{\tilde{\mathbf{Y}}_i}(1-p'+p's_1,1-p'+p's_2,1-p'+p's_3)\quad(i \in \{1,2\}).
\]
Note the contrast with the forward branching process $\mathcal{B}_F^{(2)}$, in which each of $u$'s children is infected independently with probability that depends on their type (e.g.\ \eqref{XYtrans} and \eqref{XYtrans2b}). This underlines why the model with polarized partial immunity needs to be treated separately in the backward branching process $\mathcal{B}_B^{(2)}$ but not in the forward branching process $\mathcal{B}_F^{(2)}$.
If $u$ is of type $3$, then it is necessarily susceptible to the second epidemic as it was not infected by the
first epidemic. Again each of $u$'s children in $\calTst$  contacts $u$ in the second epidemic independently with probability $p$, whence
\[
\hat{f}_{3}(\mathbf{s})=f_{\tilde{\mathbf{Y}}_3}(1-p'+p's_1,1-p'+p's_2,1-p'+p's_3).
\]

Similar arguments for the initial generation show that, using~\eqref{equ:tower1},
\[
\check{f}_1(\mathbf{s})=1-\alpha+\alpha f_{\mathbf{Y}_1}(1-p+ps_1,1-p+ps_2,1-p+ps_3)
\]
and
\[
\check{f}_3(\mathbf{s})=f_{\mathbf{Y}_3}(1-p+ps_1,1-p+ps_2,1-p+ps_3).
\]

These all agree with the offspring PGFs given in the statement of Theorem~\ref{thmpol}, thus completing its proof.

\deleted{In contrast to the forward branching process, calculation of the conditional expectations $\mathbb{E}[s_{1}^{\hat{X}_{i1} }s_{2}^{\hat{X}_{i2} } s_{3}^{\hat{X}_{i3} }|\tilde{Y}_{i1},\tilde{Y}_{i2},\tilde{Y}_{i3}]$ and $\mathbb{E}[s_{1}^{\check{X}_{i1} }s_{2}^{\check{X}_{i2} } s_{3}^{\check{X}_{i3} }|
Y_{i1},Y_{i2},Y_{i3}]$ depends on the model used for partial immunity, and is given for polarized partial immunity and leaky partial immunity in Sections~\ref{sec:polfinal} and~\ref{leaky:final}, respectively.}

\section{Discussion}
\label{sec:disc}

In this paper we derive results for two sequential SIR epidemics on the same (static) configuration model network. Essential in our analysis is that the first epidemic is completely ended before the second epidemic starts. Furthermore, we assume that the infectious period of
 individuals is not random.

We think that it is not too hard to obtain similar results to those in this paper if we assume that the infectious periods of infected individuals are independent and identically distributed. However, we also recognise that it is not \replaced{trivial}{entirely straightforward}. In particular, we will need a four-type branching process to approximate the epidemic instead of a three-type branching process in order to deal with dependencies which arise because whether or not two neighbours of a vertex ($v$ say) are infected by the first epidemic both depend on the infectious period of $v$ \citep{Kuul82,Ball09,Kenah11,Meester}. To deal with random infectious periods, we need to consider $\vec{\calG}$, the directed version of $\calG$ in which all edges of $\calG$ are replaced by two directed edges in opposite directions.
\replaced
{The directed edges from an individual, $u$ say, in $\vec{\calG}$ are coloured red independently given $u$'s infectious period $I_u$, but with a probability that depends on $I_u$.}{The directed edges of $\vec{\calG}$ are now red with a probability (only) depending on the infectious period of the starting vertex of the edge.} Again we can approximate $\vec{\calG}$ by a random tree ($\vec{\calT}$ say), rooted at $w_0$, which is the initial infected individual for the second epidemic. We then colour a vertex in $\vec{\calT}$ red if there is an infinite red path of edges directed towards it.
Now assume that a given vertex, $v$ say, is red because there is an infinite red path towards it, which does not involve the parent of $v$ (so there is a red path of descent towards $v$ in $\vec{\calT}$). \added{The main complication, not immediately solved by using just a directed network, is the following.}
Assume, in the terminology of Section \ref{branproc}, that $v$ has a child which has no infinite red path from above towards it. Then whether or not this child is red is dependent on whether or not the parent of $v$ is red, through the infectious period of $v$.
We can deal with this dependency by assigning subtypes to type-1 vertices, in which vertex $v$ is of type 1a if it has an infinite red path from above pointing towards it and there is a red edge from vertex $v$ towards its parent. Vertex $v$ is of type 1b if it has an infinite red path from above pointing towards it and there is no red edge from vertex $v$ towards its parent. Just as in the fixed infectious period case, vertex $v$ is of type 2 if it has at least one infinite red path towards it, but all those infinite red paths contain the parent of $v$ as well, while a type-3 vertex is one which does not have a red path leading towards it.
Once we have defined the four types of vertices in $\vec{\calT}$, we expect the analysis of the second epidemic to be very similar in spirit to the analysis presented in this paper. However, the notation and presentation will become more cumbersome and will obscure the main message of this paper.

In the polarized partial immunity example we consider two repeated epidemics on the same graph $\calG$, in which individuals infected in the first epidemic are susceptible to the second epidemic with probability $\alpha$ independently of one another. We might ask what happens if we we have more than two epidemics and individuals which were either immune to the $k$-th epidemic ($k \in \mathbb{N}$), or were infected in the $k$-th epidemic are susceptible to the $k+1$-st epidemic with probability $\alpha$. After every epidemic the population is in a state in the state space $\{S,R\}^{\calG}$ and those states after the epidemics form a Markov chain on $\{S,R\}^{\calG}$ with a stationary distribution. It would be interesting, but probably hard, to analyse properties of this stationary distribution.

  \backmatter

\bmhead{Acknowledgments}

This work was also supported by a grant from the Knut and Alice Wallenberg Foundation, which enabled the first author to be a guest professor at the Department of Mathematics, Stockholm University.


\begin{thebibliography}{27}
\ifx \bisbn   \undefined \def \bisbn  #1{ISBN #1}\fi
\ifx \binits  \undefined \def \binits#1{#1}\fi
\ifx \bauthor  \undefined \def \bauthor#1{#1}\fi
\ifx \batitle  \undefined \def \batitle#1{#1}\fi
\ifx \bjtitle  \undefined \def \bjtitle#1{#1}\fi
\ifx \bvolume  \undefined \def \bvolume#1{\textbf{#1}}\fi
\ifx \byear  \undefined \def \byear#1{#1}\fi
\ifx \bissue  \undefined \def \bissue#1{#1}\fi
\ifx \bfpage  \undefined \def \bfpage#1{#1}\fi
\ifx \blpage  \undefined \def \blpage #1{#1}\fi
\ifx \burl  \undefined \def \burl#1{\textsf{#1}}\fi
\ifx \doiurl  \undefined \def \doiurl#1{\url{https://doi.org/#1}}\fi
\ifx \betal  \undefined \def \betal{\textit{et al.}}\fi
\ifx \binstitute  \undefined \def \binstitute#1{#1}\fi
\ifx \binstitutionaled  \undefined \def \binstitutionaled#1{#1}\fi
\ifx \bctitle  \undefined \def \bctitle#1{#1}\fi
\ifx \beditor  \undefined \def \beditor#1{#1}\fi
\ifx \bpublisher  \undefined \def \bpublisher#1{#1}\fi
\ifx \bbtitle  \undefined \def \bbtitle#1{#1}\fi
\ifx \bedition  \undefined \def \bedition#1{#1}\fi
\ifx \bseriesno  \undefined \def \bseriesno#1{#1}\fi
\ifx \blocation  \undefined \def \blocation#1{#1}\fi
\ifx \bsertitle  \undefined \def \bsertitle#1{#1}\fi
\ifx \bsnm \undefined \def \bsnm#1{#1}\fi
\ifx \bsuffix \undefined \def \bsuffix#1{#1}\fi
\ifx \bparticle \undefined \def \bparticle#1{#1}\fi
\ifx \barticle \undefined \def \barticle#1{#1}\fi
\bibcommenthead
\ifx \bconfdate \undefined \def \bconfdate #1{#1}\fi
\ifx \botherref \undefined \def \botherref #1{#1}\fi
\ifx \url \undefined \def \url#1{\textsf{#1}}\fi
\ifx \bchapter \undefined \def \bchapter#1{#1}\fi
\ifx \bbook \undefined \def \bbook#1{#1}\fi
\ifx \bcomment \undefined \def \bcomment#1{#1}\fi
\ifx \oauthor \undefined \def \oauthor#1{#1}\fi
\ifx \citeauthoryear \undefined \def \citeauthoryear#1{#1}\fi
\ifx \endbibitem  \undefined \def \endbibitem {}\fi
\ifx \bconflocation  \undefined \def \bconflocation#1{#1}\fi
\ifx \arxivurl  \undefined \def \arxivurl#1{\textsf{#1}}\fi
\csname PreBibitemsHook\endcsname




\bibitem[\protect\citeauthoryear{Andersson}{1999}]{Ande99}
\begin{barticle}
\bauthor{\bsnm{Andersson}, \binits{H.}}:
\batitle{Epidemic models and social networks}.
\bjtitle{Mathematical {S}cientist}
\bvolume{24}(\bissue{2}),
\bfpage{128}--\blpage{147}
(\byear{1999})
\end{barticle}
\endbibitem

\bibitem[\protect\citeauthoryear{Andersson and
  Britton}{2000}]{Ande00}
\begin{bbook}
\bauthor{\bsnm{Andersson}, \binits{H.}},
\bauthor{\bsnm{Britton}, \binits{T.}}:
\bbtitle{Stochastic {E}pidemic {M}odels and {T}heir {S}tatistical {A}nalysis}
vol. \bseriesno{151}.
\bpublisher{Springer}, 
(\byear{2000})
\end{bbook}
\endbibitem


\bibitem[\protect\citeauthoryear{Ball and Neal}{2002}]{Ball02}
\begin{barticle}
\bauthor{\bsnm{Ball}, \binits{F.}},
\bauthor{\bsnm{Neal}, \binits{P.}}:
\batitle{A general model for stochastic SIR epidemics with two levels of mixing}.
\bjtitle{Mathematical Biosciences}
\bvolume{180}(\bissue{1--2}),
\bfpage{73}--\blpage{102}
(\byear{2002})
\end{barticle}
\endbibitem

\bibitem[\protect\citeauthoryear{Ball and Lyne}{2006}]{BallLyne06}
\begin{barticle}
\bauthor{\bsnm{Ball}, \binits{F.}},
\bauthor{\bsnm{Lyne}, \binits{O.}}:
\batitle{Optimal vaccination schemes for epidemics among a population of
  households, with application to variola minor in {B}razil}.
\bjtitle{Statistical Methods in Medical Research}
\bvolume{15}(\bissue{5}),
\bfpage{481}--\blpage{497}
(\byear{2006})
\end{barticle}
\endbibitem

\bibitem[\protect\citeauthoryear{Ball et~al.}{2009}]{Ball09}
\begin{barticle}
\bauthor{\bsnm{Ball}, \binits{F.}},
\bauthor{\bsnm{Sirl}, \binits{D.}},
\bauthor{\bsnm{Trapman}, \binits{P.}}:
\batitle{Threshold behaviour and final outcome of an epidemic on a random
  network with household structure}.
\bjtitle{Advances in Applied Probability}
\bvolume{41}(\bissue{3}),
\bfpage{765}--\blpage{796}
(\byear{2009})
\end{barticle}
\endbibitem

\bibitem[\protect\citeauthoryear{Ball et~al.}{2014}]{ball_sirl_trapman_2014}
\begin{barticle}
\bauthor{\bsnm{Ball}, \binits{F.}},
\bauthor{\bsnm{Sirl}, \binits{D.}},
\bauthor{\bsnm{Trapman}, \binits{P.}}:
\batitle{Epidemics on random intersection graphs}.
\bjtitle{The Annals of Applied Probability}
\bvolume{24}(\bissue{3}),
\bfpage{1081}--\blpage{1128}
(\byear{2014})
\end{barticle}
\endbibitem

\bibitem[\protect\citeauthoryear{Ball}{2019}]{Ball19}
\begin{barticle}
\bauthor{\bsnm{Ball}, \binits{F.}}:
\batitle{Susceptibility sets and the final outcome of collective {R}eed-{F}rost
  epidemics}.
\bjtitle{Methodology and Computing in Applied Probability}
\bvolume{21}(\bissue{2}),
\bfpage{401}--\blpage{421}
(\byear{2019})
\end{barticle}
\endbibitem



\bibitem[\protect\citeauthoryear{Bansal and Meyers}{2012}]{Bans12}
\begin{barticle}
\bauthor{\bsnm{Bansal}, \binits{S.}},
\bauthor{\bsnm{Meyers}, \binits{L.A.}}:
\batitle{The impact of past epidemics on future disease dynamics}.
\bjtitle{Journal of Theoretical Biology}
\bvolume{309},
\bfpage{176}--\blpage{184}
(\byear{2012})
\end{barticle}
\endbibitem

\bibitem[\protect\citeauthoryear{Becker and Starczak}{1998}]{BeckerStarczak98}
\begin{barticle}
\bauthor{\bsnm{Becker}, \binits{N.G.}},
\bauthor{\bsnm{Starczak}, \binits{D.N.}}:
\batitle{The effect of random vaccine response on the vaccination coverage
  required to prevent epidemics}.
\bjtitle{Mathematical Biosciences}
\bvolume{154}(\bissue{2}),
\bfpage{117}--\blpage{135}
(\byear{1998})
\end{barticle}
\endbibitem

\bibitem[\protect\citeauthoryear{Britton}{2010}]{britton2010stochastic}
\begin{barticle}
\bauthor{\bsnm{Britton}, \binits{T.}}:
\batitle{Stochastic epidemic models: a survey}.
\bjtitle{Mathematical Biosciences}
\bvolume{225}(\bissue{1}),
\bfpage{24}--\blpage{35}
(\byear{2010})
\end{barticle}
\endbibitem

\bibitem[\protect\citeauthoryear{Britton et~al.}{2019}]{Brit19}
\begin{barticle}
\bauthor{\bsnm{Britton}, \binits{T.}},
\bauthor{\bsnm{Leung}, \binits{K.Y.}},
\bauthor{\bsnm{Trapman}, \binits{P.}}:
\batitle{Who is the infector? general multi-type epidemics and real-time
  susceptibility processes}.
\bjtitle{Advances in Applied Probability}
\bvolume{51}(\bissue{2}),
\bfpage{606}--\blpage{631}
(\byear{2019})
\end{barticle}
\endbibitem

\bibitem[\protect\citeauthoryear{Britton et~al.}{2020}]{Brit20}
\begin{barticle}
\bauthor{\bsnm{Britton}, \binits{T.}},
\bauthor{\bsnm{Ball}, \binits{F.}},
\bauthor{\bsnm{Trapman}, \binits{P.}}:
\batitle{A mathematical model reveals the influence of population heterogeneity
  on herd immunity to {SARS-CoV}-2}.
\bjtitle{Science}
\bvolume{369}(\bissue{6505}),
\bfpage{846}--\blpage{849}
(\byear{2020})
\end{barticle}
\endbibitem

\bibitem[\protect\citeauthoryear{Cox and Durrett}{1988}]{cox1988limit}
\begin{barticle}
\bauthor{\bsnm{Cox}, \binits{J.}},
\bauthor{\bsnm{Durrett}, \binits{R.}}:
\batitle{Limit theorems for the spread of epidemics and forest fires}.
\bjtitle{Stochastic Processes and their Applications}
\bvolume{30}(\bissue{2}),
\bfpage{171}--\blpage{191}
(\byear{1988})
\end{barticle}
\endbibitem

\bibitem[\protect\citeauthoryear{Diekmann
  et~al.}{2012}]{diekmann2012mathematical}
\begin{bbook}
\bauthor{\bsnm{Diekmann}, \binits{O.}},
\bauthor{\bsnm{Heesterbeek}, \binits{H.}},
\bauthor{\bsnm{Britton}, \binits{T.}}:
\bbtitle{Mathematical {T}ools for {U}nderstanding {I}nfectious {D}isease
  {D}ynamics}.
\bpublisher{Princeton University Press}, 
(\byear{2012})
\end{bbook}
\endbibitem

\bibitem[\protect\citeauthoryear{Durrett}{2007}]{Durr07}
\begin{bbook}
\bauthor{\bsnm{Durrett}, \binits{R.}}:
\bbtitle{Random {G}raph {D}ynamics}.
\bpublisher{Cambridge University Press, Cambridge}, 
(\byear{2007})
\end{bbook}
\endbibitem

\bibitem[\protect\citeauthoryear{Funk and Jansen}{2010}]{Funk10}
\begin{barticle}
\bauthor{\bsnm{Funk}, \binits{S.}},
\bauthor{\bsnm{Jansen}, \binits{V.A.}}:
\batitle{Interacting epidemics on overlay networks}.
\bjtitle{Physical Review E}
\bvolume{81}(\bissue{3}),
\bfpage{036118}
(\byear{2010})
\end{barticle}
\endbibitem

\bibitem[\protect\citeauthoryear{Grimmett}{1999}]{Grimmett1999}
\begin{bbook}
\bauthor{\bsnm{Grimmett}, \binits{G.}}:
\bbtitle{Percolation},
\bedition{2}nd edn.
\bpublisher{Springer}, 
(\byear{1999})
\end{bbook}
\endbibitem

\bibitem[\protect\citeauthoryear{Hardy et~al.}{1952}]{Hardy1952}
\begin{bbook}
\bauthor{\bsnm{Hardy}, \binits{G.H.}},
\bauthor{\bsnm{Littlewood}, \binits{J.E.}},
\bauthor{\bsnm{P{\'o}lya}, \binits{G.}}:
\bbtitle{Inequalities}.
\bpublisher{Cambridge University Press}, \blocation{Cambridge}
(\byear{1952})
\end{bbook}
\endbibitem

\bibitem[\protect\citeauthoryear{{\noopsort{Hofstad}}{van der
  Hofstad}}{2016}]{van2016random}
\begin{bbook}
\bauthor{\bsnm{{\noopsort{Hofstad}}{van der Hofstad}}, \binits{R.}}:
\bbtitle{Random {G}raphs and {C}omplex {N}etworks}
vol. \bseriesno{1}.
\bpublisher{Cambridge Series in Statistical and Probabilistic Mathematics},
  \blocation{Cambridge}
(\byear{2016})
\end{bbook}
\endbibitem


\bibitem[\protect\citeauthoryear{Jacobsen et~al.}{2018}]{Jaco18}
\begin{barticle}
\bauthor{\bsnm{Jacobsen}, \binits{K.A.}},
\bauthor{\bsnm{Burch}, \binits{M.G.}},
\bauthor{\bsnm{Tien}, \binits{J.H.}},
\bauthor{\bsnm{Rempa{\l}a}, \binits{G.A.}}:
\batitle{The large graph limit of a stochastic epidemic model on a dynamic
  multilayer network}.
\bjtitle{Journal of Biological Dynamics}
\bvolume{12}(\bissue{1}),
\bfpage{746}--\blpage{788}
(\byear{2018})
\end{barticle}
\endbibitem



\bibitem[\protect\citeauthoryear{Jagers}{1975}]{jagers1975branching}
\begin{bbook}
\bauthor{\bsnm{Jagers}, \binits{P.}}:
\bbtitle{Branching Processes with Biological Applications}.
\bpublisher{Wiley}, 
(\byear{1975})
\end{bbook}
\endbibitem

\bibitem[\protect\citeauthoryear{Janson}{2014}]{Jans14}
\begin{barticle}
\bauthor{\bsnm{Janson}, \binits{S.}}:
\batitle{The probability that a random multigraph is simple. ii}.
\bjtitle{Journal of Applied Probability}
\bvolume{51}(\bissue{A}),
\bfpage{123}--\blpage{137}
(\byear{2014})
\end{barticle}
\endbibitem

\bibitem[\protect\citeauthoryear{Kenah and Robins}{2007}]{Kenah07}
\begin{barticle}
\bauthor{\bsnm{Kenah}, \binits{E.}},
\bauthor{\bsnm{Robins}, \binits{J.M.}}:
\batitle{Second look at the spread of epidemics on networks}.
\bjtitle{Physical Review E—Statistical, Nonlinear, and Soft Matter Physics}
\bvolume{76}(\bissue{3}),
\bfpage{036113}
(\byear{2007})
\end{barticle}
\endbibitem


\bibitem[\protect\citeauthoryear{Kenah and Robins}{2007b}]{Kenah07b}
\begin{barticle}
\bauthor{\bsnm{Kenah}, \binits{E.}},
\bauthor{\bsnm{Robins}, \binits{J.M.}}:
\batitle{Network-based analysis of stochastic SIR epidemic models with random and proportionate mixing}.
\bjtitle{Journal of theoretical biology}
\bvolume{249}(\bissue{4}),
\bfpage{706--722}
(\byear{2007})
\end{barticle}
\endbibitem

\bibitem[\protect\citeauthoryear{Kenah and Miller}{2011}]{Kenah11}
\begin{barticle}
\bauthor{\bsnm{Kenah}, \binits{E.}},
\bauthor{\bsnm{Miller}, \binits{J.C.}}:
\batitle{Epidemic percolation networks, epidemic outcomes, and interventions}.
\bjtitle{Interdisciplinary perspectives on infectious diseases}
\bvolume{2011}(\bissue{1}),
\bfpage{543520}
(\byear{2011})
\end{barticle}
\endbibitem


\bibitem[\protect\citeauthoryear{Kuulasmaa}{1982}]{Kuul82}
\begin{barticle}
\bauthor{\bsnm{Kuulasmaa}, \binits{K.}}:
\batitle{The spatial general epidemic and locally dependent random graphs}.
\bjtitle{Journal of Applied Probability}
\bvolume{19}(\bissue{4}),
\bfpage{745}--\blpage{758}
(\byear{1982})
\end{barticle}
\endbibitem

\bibitem[\protect\citeauthoryear{Lashari}{2019}]{Lash19}
\begin{bbook}
\bauthor{\bsnm{Lashari}, \binits{A.A.}}:
\bbtitle{Stochastic epidemics on random networks}.
\bpublisher{Department of Mathematics, Stockholm University},
(\byear{2019})
\end{bbook}
\endbibitem


\bibitem[\protect\citeauthoryear{Meester and Trapman}{2011}]{Meester}
\begin{barticle}
\bauthor{\bsnm{Meester}, \binits{R.}},
\bauthor{\bsnm{Trapman}, \binits{P.}}:
\batitle{Bounding basic characteristics of spatial epidemics with a new
  percolation model}.
\bjtitle{Advances in Applied Probability}
\bvolume{43}(\bissue{2}),
\bfpage{335}--\blpage{347}
(\byear{2011})
\end{barticle}
\endbibitem

\bibitem[\protect\citeauthoryear{Miller}{2007}]{Mill07}
\begin{barticle}
\bauthor{\bsnm{Miller}, \binits{J.C.}}:
\batitle{Epidemic size and probability in populations with heterogeneous infectivity and susceptibility}.
\bjtitle{Physical Review E—Statistical, Nonlinear, and Soft Matter Physics}
\bvolume{76}(\bissue{1}),
\bfpage{010101}
(\byear{2007})
\end{barticle}
\endbibitem

\bibitem[\protect\citeauthoryear{Moore et~al.}{2018}]{Moore}
\begin{barticle}
\bauthor{\bsnm{Moore}, \binits{S.}},
\bauthor{\bsnm{M{\"o}rters}, \binits{P.}},
\bauthor{\bsnm{Rogers}, \binits{T.}}:
\batitle{A re-entrant phase transition in the survival of secondary infections
  on networks}.
\bjtitle{Journal of Statistical Physics}
\bvolume{171}(\bissue{6}),
\bfpage{1122}--\blpage{1135}
(\byear{2018})
\end{barticle}
\endbibitem

\bibitem[\protect\citeauthoryear{Newman}{2002}]{newman2002spread}
\begin{barticle}
\bauthor{\bsnm{Newman}, \binits{M.E.J.}}:
\batitle{Spread of epidemic disease on networks}.
\bjtitle{Physical review E}
\bvolume{66}(\bissue{1}),
\bfpage{016128}
(\byear{2002})
\end{barticle}
\endbibitem

\bibitem[\protect\citeauthoryear{Newman}{2005}]{newman2005threshold}
\begin{barticle}
\bauthor{\bsnm{Newman}, \binits{M.E.J.}}:
\batitle{Threshold effects for two pathogens spreading on a network}.
\bjtitle{Physical review letters}
\bvolume{95}(\bissue{10}),
\bfpage{108701}
(\byear{2005})
\end{barticle}
\endbibitem


\bibitem[\protect\citeauthoryear{Newman and Ferrario}{2013}]{pmid23951134}
\begin{barticle}
\bauthor{\bsnm{Newman}, \binits{M.E.J.}},
\bauthor{\bsnm{Ferrario}, \binits{C.R.}}:
\batitle{{{I}nteracting epidemics and coinfection on contact networks}}.
\bjtitle{PLoS ONE}
\bvolume{8}(\bissue{8}),
\bfpage{71321}
(\byear{2013})
\end{barticle}
\endbibitem



\end{thebibliography}
\newcommand{\noopsort}[1]{} \newcommand{\singleletter}[1]{#1}

\appendix

\section{Appendix}
In this appendix we further describe the approximations made in \cite{Bans12} and relate them to our model and notation.

\subsection{Computations of reproduction numbers in the Bansal-Meyers approximation for polarized partial immunity}\label{App1}

Although \cite{Bans12} do not discuss the basic reproduction number itself, their equation (1) may be used to compute a basic reproduction number for the second epidemic in the polarized partial immunity case using their approximation.  We denote this quantity by $R_0^{\text{BM}}(\alpha)$.

Recall the definition of $\tilde{q}$ in \eqref{qdef} and  that $\alpha$ is the probability that an individual infected in the first epidemic is fully susceptible to the second epidemic, while $p$ is the probability that an infected vertex makes an infectious contact with a given neighbour in the first epidemic, as well as the corresponding probability for the second epidemic if the neighbour is susceptible. Finally recall that  $\mu_D = \sum_{k=0}^{\infty} k p_k$ is the expectation of the degree distribution of the network. In \cite{Bans12}, the notation for $\tilde{q}$ is $u_1$, while $p$ is denoted by $T_1$.
Define $\nu = \tilde{q} + (1-\tilde{q}) \alpha$   and  $\sigma =\tilde{q}(1-p)(1-\alpha) + \alpha$.

\cite{Bans12} then consider $p_2(\ell\mid k,\text{uninfected})$,  the probability that a vertex of degree $k$ which escaped the first epidemic has $\ell$ edges to vertices that are susceptible at the start of the second epidemic, and $p_2(\ell\mid k,\text{re-susceptible})$, the probability that a vertex of degree $k$ which did not escape the first epidemic, but became susceptible again, has $\ell$ edges to vertices that are susceptible at the start of the second epidemic. For notational convenience, below we denote ``uninfected'' by ``$\text{U}$'' and ``re-susceptible'' by ``$\text{Re}$''.
They show in the Supplementary Information for their paper that  for $k \geq 1$ and $\ell =0,1, \cdots, k$,
 $$p_2(\ell\mid k,\text{U}) = {k \choose \ell}  \nu^{\ell} (1-\nu)^{k-\ell}$$
 and
 $$p_2(\ell\mid k,\text{Re})= \alpha  {k-1 \choose \ell-1}  \sigma^{\ell-1} (1-\sigma)^{k-\ell} + (1-\alpha) {k-1 \choose \ell}  \sigma^{\ell} (1-\sigma)^{k-1-\ell},$$
 where ${k-1 \choose -1}=0$.

This gives that for $k \geq 1$
\begin{eqnarray}
\sum_{\ell=0}^k \ell (\ell-1) p_2(\ell\mid k,\text{U})  &\!=\!& k(k-1) \nu^2
,\label{app1}\\
\sum_{\ell=0}^k l\ell p_2(\ell\mid k,\text{U}) &\!=\!& k \nu,\\
\sum_{\ell=0}^k \ell (\ell-1) p_2(\ell\mid k,\text{Re})  &\!=\!& k(k-1) \sigma^2 - 2 (k-1) \sigma (\sigma-\alpha),\\
\sum_{\ell=0}^k \ell p_2(\ell\mid k,\text{Re}) &\!=\!& k \sigma - (\sigma-\alpha).\label{app4}
\end{eqnarray}

Then Bansal and Meyers approximate the second epidemic by assuming that it spreads on a fully susceptible network with degree distribution defined through $p_2(\ell)$, for $\ell =0,1,\cdots$, where $p_2(\ell)$ is given in   \citep[equation (1)]{Bans12} as
$$
p_2(\ell)= \frac{\displaystyle\sum_{k=\ell}^{\infty} p_k \left[(1-p+p \tilde q)^k  p_2(\ell\mid k,\text{U}) + \alpha (1- (1-p+p \tilde q)^k) p_2(\ell \mid k,\text{Re}) \right]}{\displaystyle\sum_{k=0}^{\infty}p_k  \left[(1-p+p \tilde q)^k  + \alpha (1- (1-p+p \tilde q)^k) \right]}.
$$

The basic reproduction number according to this method then becomes after some algebra
\begin{multline*}
R_0^{\text{BM}}(\alpha)  =   p \frac{\displaystyle\sum_{\ell=0}^{\infty} \ell(\ell-1)p_2(\ell)}{\displaystyle\sum_{\ell=0}^{\infty} \ell p_2(\ell)}\\
 =  p \frac{\displaystyle\sum_{k=0}^{\infty} p_k\displaystyle\sum_{\ell=0}^k \ell(\ell-1)\left[(1-p+p \tilde q)^k  p_2(\ell|k,\text{U}) + \alpha (1- (1-p+p \tilde q)^k) p_2(\ell|k,\text{Re}) \right]}{\displaystyle\sum_{k=0}^{\infty}p_k \displaystyle\sum_{\ell=0}^k \ell \left[(1-p+p \tilde q)^k  p_2(\ell|k,\text{U}) + \alpha (1- (1-p+p \tilde q)^k) p_2(\ell|k,\text{Re}) \right]}.
\end{multline*}

Filling in \eqref{app1}-\eqref{app4} into this display, we obtain that the numerator of the expression is given by
\begin{align*}
\ & \nu^2\displaystyle\sum_{k=0}^{\infty} p_k   k(k-1)(1-p+p \tilde q)^k  \\
\ & + \alpha \sigma^2 \displaystyle\sum_{k=0}^{\infty} p_k  k(k-1)(1- (1-p+p \tilde q)^k)  \\
\ & - 2 \alpha \sigma(\sigma-\alpha) \displaystyle\sum_{k=0}^{\infty} p_k  (k-1) (1- (1-p+p \tilde q)^k)\\
=\ & \mu_D(\nu^2-\alpha \sigma^2)\displaystyle\sum_{k=0}^{\infty} \tilde{p}_k   (k-1)(1-p+p \tilde q)^k \\
\ & + 2 \mu_D \alpha \sigma(\sigma-\alpha) \displaystyle\sum_{k=0}^{\infty} \tilde{p}_k  (1-p+p \tilde q)^k -
2 \alpha \sigma(\sigma-\alpha) \displaystyle\sum_{k=0}^{\infty} p_k  (1-p+p \tilde  q)^k \\
\ & + \mu_D \alpha \sigma^2 \displaystyle\sum_{k=0}^{\infty} \tilde{p}_k (k-1)
- 2 \alpha \sigma(\sigma-\alpha)(\mu_D-1)\\
=\ & \mu_D(\nu^2-\alpha \sigma^2)\displaystyle\sum_{k=0}^{\infty} \tilde{p}_k   (k-1)(1-p+p \tilde q)^k \\
\ & + 2 \mu_D \alpha \sigma(\sigma-\alpha)  \tilde{q}  (1-p+p \tilde q) -
2 \alpha \sigma(\sigma-\alpha) \displaystyle\sum_{k=0}^{\infty} p_k  (1-p+p \tilde q)^k \\
\ & + \mu_D \alpha \sigma^2 \displaystyle\sum_{k=0}^{\infty} \tilde{p}_k (k-1)
- 2 \alpha \sigma(\sigma-\alpha)(\mu_D-1)\\
=\ &  \mu_D(\nu^2-\alpha \sigma^2)\displaystyle\sum_{k=0}^{\infty} \tilde{p}_k   (k-1)(1-p+p \tilde q)^k
- 2 \alpha \sigma(\sigma-\alpha) \displaystyle\sum_{k=0}^{\infty} p_k  (1-p+p \tilde q)^k\\
\ & + \mu_D \alpha \sigma^2 \displaystyle\sum_{k=0}^{\infty} \tilde{p}_k (k-1)
- 2 \mu_D \alpha \sigma(\sigma-\alpha)  (1-\tilde{q})  (1 + p \tilde q)
+ 2 \alpha \sigma(\sigma-\alpha)\\
=\ &  \mu_D(\nu^2-\alpha \sigma^2)\displaystyle\sum_{k=0}^{\infty} \tilde{p}_k   (k-1)(1-p+p \tilde q)^k
+ 2 \alpha \sigma(\sigma-\alpha)(1-q)\\
\ & + \mu_D \alpha \sigma^2 \displaystyle\sum_{k=0}^{\infty} \tilde{p}_k (k-1)
- 2 \mu_D \alpha \sigma(\sigma-\alpha)  (1-\tilde{q})  (1 + p \tilde q).
\end{align*}
The denominator is given by

\begin{align*}
\ & \nu\displaystyle\sum_{k=0}^{\infty} p_k   k (1-p+p \tilde q)^k
+ \alpha \sigma \displaystyle\sum_{k=0}^{\infty} p_k  k(1- (1-p+p \tilde q)^k)
-  \alpha (\sigma-\alpha) \displaystyle\sum_{k=0}^{\infty} p_k  (1- (1-p+p \tilde q)^k)\\
=\ &  \mu_D (\nu-\alpha \sigma) \displaystyle\sum_{k=0}^{\infty} \tilde{p}_k   (1-p+p \tilde q)^k
+  \alpha (\sigma-\alpha) \displaystyle\sum_{k=0}^{\infty} p_k  (1-p+p \tilde q)^k
+ \mu_D \alpha \sigma  -  \alpha (\sigma-\alpha)\\
=\ & \mu_D (\nu-\alpha \sigma) \tilde{q}   (1-p+p \tilde q)
+  \alpha (\sigma-\alpha) \displaystyle\sum_{k=0}^{\infty} p_k  (1-p+p \tilde q)^k
+ \mu_D \alpha \sigma  -  \alpha (\sigma-\alpha)\\
=\ & \mu_D (\nu-\alpha \sigma) \tilde{q}   (1-p+p \tilde q)
+ \mu_D \alpha \sigma  -  \alpha (\sigma-\alpha)(1-q).
\end{align*}
This leads to
\begin{multline}\label{BMR0}
R_0^{\text{BM}}(\alpha)\\ = \frac{p \mu_D (\nu^2-\alpha \sigma^2)\tau
+ \mu_D \alpha \sigma^2 R_0^{(1)}
- 2  p\alpha \sigma (\sigma-\alpha)[\mu_D  (1-\tilde{q})  (1 + p \tilde q)
-(1-q)]}
{\mu_D (\nu-\alpha \sigma) \tilde{q}   (1-p+p \tilde q)
+ \mu_D \alpha \sigma  -  \alpha (\sigma-\alpha)(1-q)},
\end{multline}
where $\tau= \displaystyle\sum_{k=0}^{\infty} \tilde{p}_k   (k-1)(1-p+p \tilde q)^k $, while $R_0^{(1)}= p \mathbb{E}[\tilde{D}-1]$ is the basic reproduction number for the first epidemic, see \eqref{firstR0}.
We do not know of an easy simplification of these expressions.

 \subsection{Bansal-Meyers approximation for leaky partial immunity}\label{sec:appLeaky}\label{App2}
In this subsection we consider the two-type percolation approximation used by \cite{Bans12} to analyse the second epidemic under leaky partial immunity.  We consider first the joint degree distributions associated with this two-type percolation process, which are key to the approximation in \cite{Bans12}.
As stated before, these distributions are incorrect in \cite{Bans12}, so to compare the approximation in \cite{Bans12} with our asymptotically exact approximation, we derive correct joint PGFs associated with these joint degree distributions.  Although the approximating percolation process is two-type, the associated forward and backward branching processes are three-type.  We give the offspring mean matrix for the forward branching process, the dominant eigenvalue of which gives an approximation $R_{0, {\rm BM}}^{(2)}$ to the basic reproduction number $R_0^{(2)}$ for the second epidemic.  The two-type percolation process can also be used to approximate the second epidemic under polarized partial immunity.  We show that $R_{0, {\rm BM}}^{(2)}$ is the same for the models of leaky partial immunity in Section \ref{sec:leakydef} and partial polarized immunity in Section \ref{sec:poldef}, provided $\alpha=\alpha_S \alpha_I$.

In \cite{Bans12}, individuals are typed $A$ or $B$, according to whether they were uninfected or infected, respectively, by the first epidemic. Key quantities in their approximation are $p_{ij}$, the probability a type-$A$ individual chosen uniformly at random from the population has $i$ type-$A$ neighbours and $j$ type-$B$ neighbours, and $q_{ij}$, the corresponding probability for a type-$B$ individual.  The following formulae for $p_{ij}$ and $q_{ij}$, translated into our notation, are given in \cite{Bans12}:
\begin{align}
p_{ij}&=\frac{p_{i+j} (1-p+p\tilde{q})^{i+j} \binom{i+j}{i} \tilde{q}^i(1-\tilde{q})^j}{q}\quad (i,j=0,1,\dots),\label{BMLpA}\\
q_{ij}&=\frac{p_{i+j}[1-(1-p+p\tilde{q})^{i+j}] \binom{i+j-1}{i} [\tilde{q}(1-p)]^i[1-\tilde{q}(1-p)]^{j-1}}{1-q} \quad (i=0,1,\dots; j=1,2,\dots) \nonumber.
\end{align}
However, these are incorrect.  For example, if $p=1$ then all neighbours of a type-$A$ individual are necessarily of type $A$, so $p_{ij}=0$ for $j>0$, unlike in~\eqref{BMLpA}.  The error with these formulae is that they use the unconditional probability that an individual of degree $i+j$ is uninfected by the first epidemic, rather than the conditional probability given the types of its neighbours.  \comment{FGB: Brief explanation added.}

To derive the correct formula for $p_{ij}$, note that the probability that an individual, $i_*$ say, chosen uniformly at random from the population has $i$ type-$A$ and $j$ type-$B$ neighbours, and is of type $A$ is $p_{i+j} \binom{i+j}{i} \tilde{q}^i(1-\tilde{q})^j (1-p)^j$, since each of the $j$ type-$B$ neighbours fail to infect $i_*$.  The unconditional probability that $i_*$ is of type $A$ is $q$.  Thus,
\[
p_{ij}=\frac{p_{i+j} \binom{i+j}{i} \tilde{q}^i[(1-\tilde{q})(1-p)]^j}{q}\quad (i,j=0,1,\dots).
\]
In the notation of \cite{Bans12}, let $f_A(x,y)$ be the bivariate PGF of this distribution.  Then
\begin{align}
\label{fA}
f_A(x,y)&=\frac{1}{q} \sum_{i=0}^{\infty} \sum_{j=0}^{\infty} p_{i+j} \binom{i+j}{i} \tilde{q}^i[(1-\tilde{q})(1-p)]^j x^i y^j \nonumber \\
&=\frac{1}{q}\sum_{k=0}^{\infty} p_k \sum_{i=0}^{k} \binom{k}{i} (\tilde{q}x)^i[(1-\tilde{q})(1-p)y]^{k-i} \nonumber \\
&=\frac{1}{q}\sum_{k=0}^{\infty} p_k [\tilde{q}x+(1-\tilde{q})(1-p)y]^k \nonumber \\
&=\frac{1}{q}f_D(\tilde{q}x+(1-\tilde{q})(1-p)y),
\end{align}
where $f_D$ is the PGF of $D$ and we have used the binomial theorem in the penultimate line.

Rather than finding $q_{ij}$, it is simpler to split the type-$B$ neighbours of a type-B individual, $j_*$ say, chosen uniformly at random from all type-$B$ individuals, into types $B_1$ and $B_2$, where a neighbour, $k_*$ say, of $j_*$ is of type $B_1$ if one of its neighbours other than $j_*$ makes infectious contact with it, otherwise $k_*$ is of type $B_2$.  Let $q_{i,j_1,j_2}$ be the probability that $j_*$ has $i$ type-$A$, $j_1$ type-$B_1$ and $j_2$ type-$B_2$ neighbours.  Note that, unconditionally, neighbours of $j_*$ are independently of types $A, B_1$ and $B_2$ with probabilities $\tilde{q}(1-p), 1-\tilde{q}$ and $\tilde{q}p$.  However, given that $j_*$ has $j_1$ type-$B_1$ neighbours, at least one of them must make infectious contact with $j_*$ (otherwise $j_*$ would not be infected), an event that occurs with conditional probability $1-(1-p)^{j_1}$.  Thus,
\[
q_{i,j_1,j_2}=\frac{1}{1-q}p_{i+j_1+j_2}\frac{(i+j_1+j_2)!}{i!j_1!j_2!}[\tilde{q}(1-p)]^i (1-\tilde{q})^{j_1} (\tilde{q}p)^{j_2} [1-(1-p)^{j_1}].
\]
Note that this formula is valid for all $i, j_1, j_2=0,1,\dots$, as $1-(1-p)^{j_1}=0$ when $j_1=0$.  Similar manipulations to those used in the derivation of~\eqref{fA} yield
\begin{align*}
\hat{f}_B(x,y_1, y_2)&= \sum_{i=1}^{\infty}\sum_{j_1=1}^{\infty}\sum_{j_2=1}^{\infty} q_{i,j_1,j_2} x^i y_1^{j_1} y_2^{j_2}\\
&=\frac{1}{1-q}\left[f_D(\tilde{q}(1-p)x+(1-\tilde{q})y_1+\tilde{q}py_2)\right.\\
&\qquad\qquad\qquad\left.-f_D(\tilde{q}(1-p)x+(1-\tilde{q})(1-p)y_1+\tilde{q}py_2)\right].
\end{align*}
Let $f_B(x,y)$ be the bivariate PGF of the distribution given by $q_{ij}$.  Then $f_B(x,y)=\hat{f}_B(x,y,y)$, whence
\begin{align}
\label{fB}
f_B(x,y)&=\frac{1}{1-q}\left[f_D(\tilde{q}(1-p)x+[1-\tilde{q}+\tilde{q}p]y)\right.\nonumber\\
&\qquad\qquad\qquad\left.-f_D(\tilde{q}(1-p)x+[(1-\tilde{q})(1-p)+\tilde{q}p]y)\right].
\end{align}

For a function $f(x,y)$, let $\partial_{1,0}f(x,y)$ and $\partial_{0,1}f(x,y)$ denote its partial derivative with respect to $x$ and $y$, respectively.  Four further bivariate PGFs are defined in \cite{Bans12}, viz.
\[
f_{AA}(x,y)=\frac{\partial_{1,0}f_A(x,y)}{\partial_{1,0}f_A(1,1)}, \quad f_{BA}(x,y)=\frac{\partial_{0,1}f_A(x,y)}{\partial_{0,1}f_A(1,1)},
\]
\[
f_{AB}(x,y)=\frac{\partial_{1,0}f_B(x,y)}{\partial_{1,0}f_B(1,1)},\quad
f_{BB}(x,y)=\frac{\partial_{0,1}f_B(x,y)}{\partial_{0,1}f_B(1,1)}.
\]
For example, $f_{BA}(x,y)$ is the bivariate PGF of the $A$-excess degree and $B$-excess degree of a type-$A$ individual, $i_*$ say, that is reached by following a randomly chosen edge from a type-$B$ individual (i.e.~of the numbers of type-$A$ and type-$B$ neighbours of $i_*$, excluding the type-$B$ neighbour from which it was reached).  These four bivariate PGFs, which follow easily from ~\eqref{fA} and~\eqref{fB}, underpin the forward and backward branching process approximations for the second epidemic.  We have that $f_{AA}=f_{BA}$ but $f_{AB} \ne f_{BB}$, so these approximating branching processes have three types.  For the forward process the types are: type 1, a $B$ individual who is infected in the second epidemic by an $A$ individual; type 2, a $B$ individual who is infected in the second epidemic by a $B$ individual; and type 3, an $A$ individual.  The corresponding mean offspring matrix is
\begin{equation}
\label{equ:BMapproxM}
M_{\rm BM}=\begin{pmatrix}
0 & \pi_{11}\hat{\mu}_{BB}^A &\pi_{10}\hat{\mu}_{BA}^A \\
0 & \pi_{11}\hat{\mu}_{BB}^B &\pi_{10}\hat{\mu}_{BA}^B \\
\pi_{01}\hat{\mu}_{AB} & 0 & \pi_{00}\hat{\mu}_{AA}
\end{pmatrix} ,
\end{equation}
where
\begin{align*}
\hat{\mu}_{AA}&=\partial_{1,0}f_{AA}(1,1)=\tilde{q}\frac{f_D''(\delta)}{f_D'(\delta)},\\
\hat{\mu}_{AB}&=\partial_{0,1}f_{AA}(1,1)=(1-\tilde{q})(1-p)\frac{f_D''(\delta)}{f_D'(\delta)},\\
\hat{\mu}_{BA}^A&=\partial_{1,0}f_{AB}(1,1)=\frac{\tilde{q}(1-p)}{\mu_D-f_D'(\delta)}\left(f_D''(1)-f_D''(\delta)\right),\\
\hat{\mu}_{BB}^A&=\partial_{0,1}f_{AB}(1,1)=\frac{1}{\mu_D-f_D'(\delta)}\left(\delta_1 f_D''(1)-\delta_2 f_D''(\delta)\right),\\
\hat{\mu}_{BA}^B&=\partial_{1,0}f_{BB}(1,1)=\frac{\tilde{q}(1-p)}{\delta_1 \mu_D-\delta_2 f_D'(\delta)}\left(\delta_1 f_D''(1)-\delta_2 f_D''(\delta)\right),\\
\hat{\mu}_{BB}^B&=\partial_{0,1}f_{BB}(1,1)=\frac{1}{\delta_1 \mu_D-\delta_2 f_D'(\delta)}\left(\delta_1^2 f_D''(1)-\delta_2^2 f_D''(\delta)\right),
\end{align*}
with
\[
\delta= 1-p+p\tilde{q},\quad \delta_1=1-\tilde{q}+p\tilde{q} \quad\text{and}\quad \delta_2=(1-p)(1-\tilde{q})+p\tilde{q}.
\]
Under this approximation, the threshold parameter $R_{0, {\rm BM}}^{(2)}$ for the second epidemic is given by the dominant eigenvalue of $M_{\rm BM}$.

Suppose that the probabilities of transmission for pairs of neighbours in the second epidemic take the form $\pi_{00}=p$, $\pi_{01}=\alpha_S p$, $\pi_{10}=\alpha_I p$ and $\pi_{11}=\alpha_S \alpha_I p$, as in the model for leaky partial immunity in Section \ref{sec:leakydef}.
Then, as shown in Section \ref{leaky:example}, the value of $R_0^{(2)}$ is the same as that under the polarized model for partial immunity,
in which $\pi_{00}=p$, $\pi_{01}=\alpha p$, $\pi_{10}= p$ and $\pi_{11}=\alpha p$, provided $\alpha_S \alpha_I =\alpha$.  The above two-type percolation approximation used by \cite{Bans12} for leaky partial immunity can also be used for polarized partial immunity and, provided $\alpha_S \alpha_I =\alpha$, the similarity transformation used in Section \ref{leaky:example} shows that under this approximation, $R_0^{(2)}$ is again the same for leaky and polarized partial immunity.

Let $\alpha_c^{\rm BM}$ denote the corresponding critical value of $\alpha$ under this approximation.  It follows from~\eqref{equ:BMapproxM} and a little algebra that $\alpha_c^{\rm BM}$ is the smallest root in $[0, 1]$ of the quadratic equation
$$
p^3\hat{\mu}_{AB} \left(\hat{\mu}_{BB}^A \hat{\mu}_{BA}^B-\hat{\mu}_{BB}^B \hat{\mu}_{BA}^A\right)\alpha^2+p\left[\hat{\mu}_{BB}^B+p\left(\hat{\mu}_{AB}\hat{\mu}_{BA}^A-\hat{\mu}_{AA}\hat{\mu}_{BB}^B\right)\right]\alpha+p\hat{\mu}_{AA}-1=0.
$$


\end{document}